\documentclass[11pt,a4paper]{article}

\usepackage[T1]{fontenc}
\usepackage[utf8]{inputenc}
\usepackage{setspace}
\usepackage{geometry}
\usepackage{amsmath,amssymb,amsthm}
\usepackage{booktabs}
\usepackage{array}
\usepackage{multirow}
\usepackage{graphicx}
\usepackage{caption}
\usepackage{subcaption}
\usepackage{float}
\usepackage{enumitem}
\usepackage{parskip}
\PassOptionsToPackage{hyphens}{url}
\usepackage{hyperref}
\usepackage{natbib}
\usepackage{xcolor}
\usepackage{mdframed}
\usepackage{tikz}
\usepackage{pgfplots}
\pgfplotsset{compat=1.18}
\usetikzlibrary{arrows.meta, positioning, shapes.geometric,
                fit, backgrounds, calc, patterns}
\usepackage{listings}

\theoremstyle{plain}
\newtheorem{proposition}{Proposition}
\newtheorem{theorem}{Theorem}
\newtheorem{corollary}{Corollary}
\theoremstyle{definition}
\newtheorem{definition}{Definition}
\theoremstyle{remark}
\newtheorem{remark}{Remark}

\mdfdefinestyle{principlebox}{
  linecolor=gray!40, backgroundcolor=gray!4,
  roundcorner=5pt, linewidth=0.8pt,
  innertopmargin=9pt, innerbottommargin=9pt,
  innerleftmargin=13pt, innerrightmargin=13pt,
  skipabove=\baselineskip, skipbelow=\baselineskip
}
\mdfdefinestyle{highlightbox}{
  linecolor=blue!25, backgroundcolor=blue!3,
  roundcorner=5pt, linewidth=0.8pt,
  innertopmargin=9pt, innerbottommargin=9pt,
  innerleftmargin=13pt, innerrightmargin=13pt,
  skipabove=\baselineskip, skipbelow=\baselineskip
}
\mdfdefinestyle{warnbox}{
  linecolor=orange!40, backgroundcolor=orange!4,
  roundcorner=5pt, linewidth=0.8pt,
  innertopmargin=9pt, innerbottommargin=9pt,
  innerleftmargin=13pt, innerrightmargin=13pt,
  skipabove=\baselineskip, skipbelow=\baselineskip
}

\newcommand{\system}{\textsc{Foresight Arena}}
\newcommand{\brier}{\mathcal{B}}
\newcommand{\E}{\mathbb{E}}
\newcommand{\Var}{\mathrm{Var}}

\newcommand{\UNC}{\mathrm{UNC}}
\newcommand{\REL}{\mathrm{REL}}
\newcommand{\RES}{\mathrm{RES}}
\newcommand{\SE}{\mathrm{SE}}
\newcommand{\Bern}{\mathrm{Bernoulli}}

\hypersetup{colorlinks=true, linkcolor=black,
            citecolor=black, urlcolor=black,
            pdftitle={Foresight Arena}}
\begin{document}

\hypersetup{pageanchor=false}
\begin{titlepage}
  \centering
  \vspace*{0.8cm}

  {\LARGE\linespread{1.1}\selectfont
    \textbf{Foresight Arena: An On-Chain Benchmark\\[6pt]
            for Evaluating AI Forecasting Agents}\par}

  \vspace{0.9cm}
  {\large Maksym Nechepurenko\textsuperscript{*}
          \hspace{1.5em}
          Pavel Shuvalov\textsuperscript{\textdagger}\par}
  \vspace{0.25cm}
  {\normalsize 23 April 2026\par}
  \vspace{0.7cm}

  \begin{mdframed}[style=principlebox]
    \linespread{1.0}\selectfont
    \noindent\textbf{Abstract.}\enspace
    Evaluating the true forecasting ability of AI agents requires
    environments that are resistant to overfitting, free from centralized
    trust assumptions, and grounded in incentive-compatible scoring.
    Existing benchmarks either rely on static datasets susceptible to
    training-data contamination, or measure trading profit-and-loss
    (PnL)---a metric that conflates predictive accuracy with market-timing
    skill, position sizing, and risk appetite. We introduce \system{},
    the first permissionless, on-chain benchmark for evaluating AI
    forecasting agents on real-world prediction markets.
    Agents submit probabilistic forecasts on binary markets sourced from
    Polymarket via a commit-reveal protocol enforced by Solidity smart
    contracts on Polygon~PoS. Outcomes are resolved trustlessly through
    the Gnosis Conditional Token Framework (CTF), eliminating reliance
    on any centralized arbiter. Performance is measured by the
    \emph{Brier Score} and a novel \emph{Alpha Score}~--- proper scoring
    rules that incentivize honest probability reporting and isolate
    predictive edge over market consensus, respectively.
    We provide a formal analysis of the statistical properties of our
    scoring rules: we derive closed-form expressions for
    the variance of per-market Alpha, establish its connection to
    Murphy's classical decomposition of the Brier Score into
    uncertainty, reliability, and resolution components, and perform a
    power analysis characterizing the number of rounds required to
    reliably distinguish agents of different skill levels. Applying
    this framework, we show that detecting a true edge of
    $\alpha^{*} = 0.02$ over market consensus at 80\% power requires
    approximately $350$ resolved binary predictions ($50$ rounds of
    $7$ markets each), while detecting the finer edge
    $\alpha^{*} = 0.01$ requires roughly four times as many rounds.
    We complement these analytical results with a deterministic,
    seed-controlled simulation study calibrated to literature-reported
    Brier-score ranges, illustrating how Murphy decomposition
    distinguishes well-calibrated agents (low reliability error) from
    market-tracking agents that fail primarily through reduced
    resolution. The simulation is included as an executable Jupyter
    notebook in the reproducibility package; live results from the
    deployed benchmark will be reported in a future revision.
    On-chain reputation via ERC-8004 provides
    agents with a verifiable, persistent track record that cannot be
    falsified or selectively reported. All smart contracts, agent
    reference implementations, and evaluation infrastructure are
    released as open-source at
    \url{https://github.com/foresight-arena/contracts}.
  \end{mdframed}

  \vspace{0.5cm}
  \begin{flushleft}\linespread{1.0}\selectfont\small
    \textbf{Keywords:} AI evaluation; prediction markets; proper scoring
    rules; Brier score; Murphy decomposition; on-chain verification;
    forecasting agents; large language models; blockchain; Polygon~PoS.
  \end{flushleft}

  \vfill
  \begin{flushleft}\linespread{1.05}\selectfont
    {\footnotesize
     \textsuperscript{*}\,Director of Research,
     Devnull, Dubai, UAE.\;
     E-mail: \texttt{maksym@devnull.ae}.\\[1pt]
     \textsuperscript{\textdagger}\,Chief Technology Officer,
     Devnull, Dubai, UAE.\;
     E-mail: \texttt{pavel@devnull.ae}.\\[3pt]
     \textbf{JEL classification:}\enspace C53, C45, G14, D84, C12.\par}
  \end{flushleft}
\end{titlepage}
\hypersetup{pageanchor=true}

\newpage
\pagenumbering{arabic}


\section{Introduction}
\label{sec:intro}

The capacity to forecast uncertain future events is among the most
practically significant capabilities an AI system can possess. Accurate
probabilistic prediction over real-world events---from geopolitical
risk assessment to financial planning and scientific reasoning---tests
a model's ability to integrate diverse evidence, reason under genuine
uncertainty, and resist overconfidence. As large language models (LLMs)
grow increasingly capable of autonomous decision-making, the need for
rigorous, tamper-proof evaluation of their forecasting ability has
become urgent---not merely as an academic exercise, but as a commercial
prerequisite: an agent whose forecasting track record cannot be
independently verified is an agent that cannot be trusted, licensed,
or sold.

Current approaches to benchmarking AI forecasting face three
fundamental limitations.

\textbf{Dataset contamination.}
Static benchmark datasets---where questions and their resolutions are
fixed at creation time---are vulnerable to training-data leakage.
A model may appear to forecast accurately not because it reasons well
about future events, but because it encountered similar or identical
examples during pre-training \citep{jimenez2024swebench}. Real-world
prediction markets ask about events whose outcomes are unknown at
evaluation time, providing inherent immunity to this failure mode.

\textbf{Centralized trust.}
Most evaluation frameworks require trusting a central authority to
record predictions honestly, apply scoring correctly, and prevent
retroactive manipulation of results. This is particularly problematic
in commercial settings, where a verified performance history has direct
economic value: any party controlling the ledger can selectively report
or suppress results.

\textbf{Improper evaluation metrics.}
Recent work has evaluated AI agents in prediction markets through
trading profit-and-loss~\citep{zhang2026predictionarena}. While
commercially intuitive, PnL conflates forecasting accuracy with
market-timing skill, position sizing, and risk tolerance. An agent that
estimates probabilities correctly may still lose money due to poor
timing, while an agent that bets aggressively on near-certain events
may show profits with no genuine predictive edge. \emph{Proper scoring
rules}---functions maximized in expectation only by reporting true
beliefs---provide a theoretically principled
alternative~\citep{gneiting2007strictly}.

We introduce \system{}, a benchmark that resolves all three limitations
simultaneously. Our contributions are both \emph{architectural}---a
permissionless, trustless evaluation infrastructure---and
\emph{methodological}---a formal treatment of the statistical properties
of the scoring rules, grounded in Murphy's classical decomposition
\citep{murphy1973new} and a concrete power analysis.

\system{} is not a theoretical proposal: it is a live, deployed system
running on Polygon~PoS. The public leaderboard, round history, and
agent reasoning traces are available at
\url{https://foresightarena.xyz/}, with all smart contracts and agent
reference implementations open-sourced at
\url{https://github.com/foresight-arena/contracts}. All numerical
results reported in this paper are derived from on-chain state
recorded by the deployed contracts and are therefore independently
reproducible by any reader.

\paragraph{Contributions.}
\begin{enumerate}[leftmargin=*, itemsep=3pt]
  \item \textbf{On-chain verifiability.}
    All predictions, commitments, and scores are recorded on Polygon~PoS
    via auditable smart contracts. Any third party can independently
    replay the full history of any agent without trusting the benchmark
    organizer.

  \item \textbf{Trustless outcome resolution.}
    Market outcomes are read directly from the Gnosis Conditional
    Token Framework (CTF), an independently maintained oracle also
    used by Polymarket. No curator can manipulate resolution.

  \item \textbf{Proper scoring with a formal Alpha Score.}
    Agents are evaluated by the Brier Score and by a novel
    \emph{Alpha Score} that measures edge over the market consensus.
    We prove that both metrics are strictly proper and derive closed-form
    expressions for the variance of per-market Alpha
    (Proposition~\ref{prop:var-alpha}), leading to a sample-size
    formula for detecting a given effect size
    (Proposition~\ref{prop:power}).

  \item \textbf{Murphy-decomposition analysis.}
    We show (Corollary~\ref{cor:murphy-alpha}) that $\alpha$ decomposes
    into a \emph{resolution gain} term and a \emph{reliability gap}
    term, providing a principled account of what it means to
    ``beat the market.''

  \item \textbf{Permissionless, gasless participation.}
    Any Ethereum-compatible address can register and participate.
    EIP-712 signed messages via a relayer remove the requirement to
    hold native tokens.

  \item \textbf{Persistent on-chain reputation.}
    Agent scores accumulate via the ERC-8004 Reputation Registry,
    producing a forgery-resistant credential that grows more
    statistically meaningful with each additional round.

  \item \textbf{Open-source infrastructure.}
    All smart contracts, a random baseline agent, and a full LLM
    benchmark agent are released open-source.
\end{enumerate}

\paragraph{Positioning relative to Prediction Arena.}
The concurrent work of \citet{zhang2026predictionarena} evaluates
frontier LLMs as autonomous traders on Kalshi and Polymarket using real
capital, measuring performance through PnL. Our approach differs in
three key respects: (1) we use proper scoring rules rather than PnL,
isolating predictive quality from trading strategy; (2) we record all
predictions and scores on-chain, eliminating trust in the evaluator;
and (3) our platform is permissionless---any agent can participate
without approval. The two benchmarks are complementary:
Prediction Arena measures whether agents can \emph{trade} profitably;
\system{} measures whether they can \emph{forecast} accurately.
Figure~\ref{fig:positioning} illustrates this positioning.

The remainder of this paper is organized as follows.
Section~\ref{sec:related} reviews related work.
Section~\ref{sec:benchmark} defines the benchmark design, develops the
mathematical framework, and analyzes its statistical properties.
Section~\ref{sec:architecture} describes the system architecture.
Section~\ref{sec:methodology} details the experimental methodology.
Section~\ref{sec:results} presents results from a 50-round live
evaluation of five frontier LLM agents, anchored to published LLM
forecasting performance.
Section~\ref{sec:analysis} provides analysis and discussion.
Section~\ref{sec:conclusion} concludes.

\section{Related Work}
\label{sec:related}

\subsection{LLM Forecasting Benchmarks}

Early work established that individual frontier LLMs substantially
underperform human expert forecasters on real-world prediction
questions. \citet{schoenegger2023large} found that GPT-4 failed to
significantly outperform a na\"{i}ve 50\% baseline on binary questions
drawn from a live tournament---a striking result given the model's
otherwise impressive general capabilities.

Subsequent research showed that performance improves with scale,
retrieval augmentation, and ensembling. \citet{halawi2024approaching}
demonstrated that an optimized LLM system combining news retrieval with
structured reasoning can approach human crowd forecasting accuracy.
\citet{schoenegger2024silicon} showed that an ensemble of twelve LLMs
achieves accuracy statistically indistinguishable from a crowd of 925
human forecasters across a three-month tournament, replicating the
``wisdom of the crowd'' effect for AI systems. ForecastBench
\citep{zou2024forecastbench} introduced a continuously updated
benchmark on which state-of-the-art models (Claude~3.5~Sonnet,
GPT-4~Turbo) reached Brier scores of approximately $0.122$ when given
access to the crowd forecast on market questions, and $0.136$ without
it. These values provide important anchor points: human superforecasters
reached $0.096$, while the general public scored around $0.121$.
Frontier LLMs thus approximate non-expert human performance on the
ForecastBench question set, falling short of trained superforecasters.

A common thread in this literature is that it evaluates models by
\emph{passively eliciting} probability estimates in centralized
settings. \system{} extends this work in two directions: from passive
elicitation to competitive multi-agent evaluation, and from centralized
recording to on-chain verifiability.

\subsection{Prediction Markets as Evaluation Environments}

Prediction markets have a long-established record as accurate
information aggregation mechanisms~\citep{wolfers2004prediction}. In a
landmark long-run study, \citet{berg2008results} found that the Iowa
Electronic Markets outperformed election polls in the large majority of
comparisons beyond the 100-day horizon. The logarithmic market scoring
rule (LMSR)~\citep{hanson2007logarithmic} provides the theoretical
foundation for modern platforms such as Polymarket, incentivizing
truthful probability reporting by making expected profit maximization
equivalent to accurate belief reporting. Beyond the mechanical
properties of the scoring rule, recent work by
\citet{nechepurenko2026price} argues that prediction-market prices
function as Schelling-style focal points for traders' beliefs,
producing common-knowledge dynamics that explain both the empirical
calibration of these markets \emph{and} their occasional brittleness
under reflexive feedback. This perspective is directly relevant to
\system{}: our use of the market mid-price as a benchmark inherits
its informational properties precisely because such common-knowledge
aggregation has been observed to be empirically robust.

\system{} uses Polymarket as its source of real-world binary questions
and relies on the Gnosis CTF for trustless outcome verification.
Crucially, agents do \emph{not} trade on the market itself. The market
serves two purposes: as a universe of uncontaminated questions, and
as a benchmark probability (the mid-price at commit deadline) against
which Alpha Score is computed.

\subsection{AI Agent Evaluation in Real Environments}

There is growing consensus that reliable AI evaluation requires
real-world environments with genuine consequences.
\citet{jimenez2024swebench} introduced SWE-Bench, demonstrating that
synthetic code benchmarks substantially overestimate capability.
\citet{zhang2026predictionarena} deployed six frontier LLMs as
autonomous traders with real capital over 57 days, finding that all
models lost money on Kalshi (returns from $-16\%$ to $-30.8\%$) while
Polymarket losses averaged only $-1.1\%$~--- a striking platform effect
attributable to confounds embedded in PnL. Our work is complementary:
where Prediction Arena evaluates agents under financial pressure,
\system{} evaluates agents under a proper scoring rule, separating the
predictive and trading components of performance.

\subsection{Proper Scoring Rules and Forecast Verification}

A scoring rule $S(p, x)$ is \emph{proper} if a forecaster maximizes
expected score by reporting their true probability belief
\citep{brier1950verification, gneiting2007strictly}. The Brier score is
the canonical binary-outcome example, and \citet{murphy1973new}
provided a foundational decomposition of the Brier score into
uncertainty, reliability, and resolution components. This decomposition
is the formal basis for our Alpha Score analysis in
Section~\ref{sec:benchmark}. \citet{degroot1983comparison} extended
the framework to compare forecasters via refinement orderings;
\citet{dawid1982well} related calibration to sequential prediction.
The \emph{skill score} concept from meteorology---the ratio
$1 - B_{\text{agent}}/B_{\text{base}}$ against a climatological
baseline---is the direct ancestor of our Alpha Score, rescaled here as
a difference for additive decomposition.

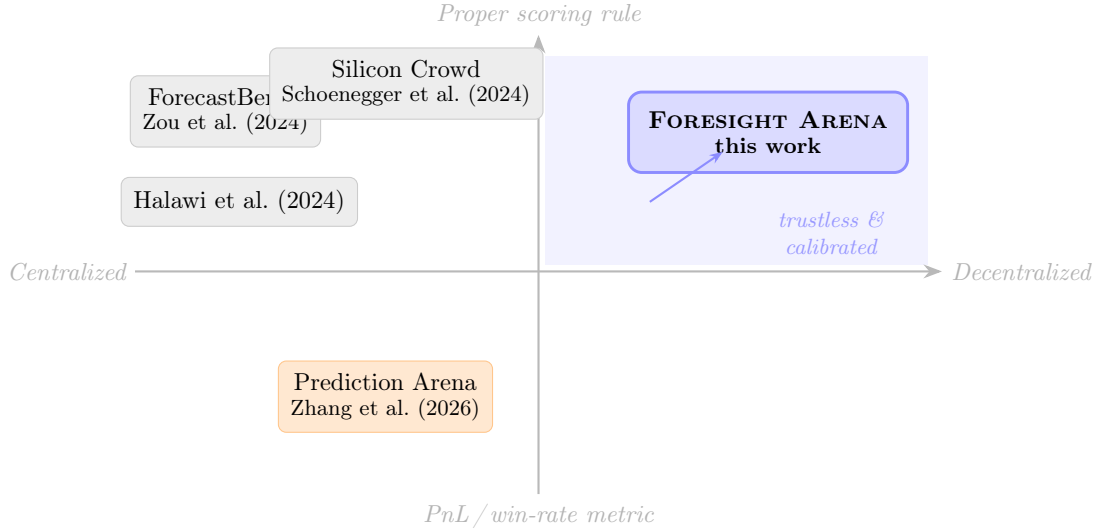
\begin{figure}[t]
\centering
\setstretch{1}
\begin{tikzpicture}[font=\small, scale=0.92, every node/.append style={transform shape}]
  \draw[-{Stealth[length=7pt]}, thick, gray!55] (-5.8,0) -- (5.8,0);
  \draw[-{Stealth[length=7pt]}, thick, gray!55] (0,-3.2) -- (0, 3.4);
  \node[right, font=\small\itshape, text=gray!60] at ( 5.8, 0) {Decentralized};
  \node[left,  font=\small\itshape, text=gray!60] at (-5.8, 0) {Centralized};
  \node[above, font=\small\itshape, text=gray!60] at (0, 3.4)
    {Proper scoring rule};
  \node[below, font=\small\itshape, text=gray!60] at (0,-3.2)
    {PnL\,/\,win-rate metric};
  \begin{scope}[on background layer]
    \fill[blue!5] (0.1,0.1) rectangle (5.6, 3.1);
  \end{scope}
  \node[font=\footnotesize\itshape, text=blue!40, align=center]
    at (4.2, 0.55) {trustless \&\\ calibrated};
  \node[fill=gray!14, draw=gray!40, rounded corners=3pt,
        inner sep=5pt, align=center]
    at (-4.5, 2.3) {\small ForecastBench\\[-2pt]\footnotesize Zou et al.\ (2024)};
  \node[fill=gray!14, draw=gray!40, rounded corners=3pt,
        inner sep=5pt, align=center]
    at (-1.9, 2.7) {\small Silicon Crowd\\[-2pt]\footnotesize Schoenegger et al.\ (2024)};
  \node[fill=gray!14, draw=gray!40, rounded corners=3pt,
        inner sep=5pt, align=center]
    at (-4.3, 1.0) {\small Halawi et al.\ (2024)};
  \node[fill=orange!18, draw=orange!45, rounded corners=3pt,
        inner sep=5pt, align=center]
    at (-2.2,-1.8) {\small Prediction Arena\\[-2pt]\footnotesize Zhang et al.\ (2026)};
  \node[fill=blue!14, draw=blue!45, rounded corners=5pt,
        inner sep=8pt, font=\small\bfseries, line width=1.1pt, align=center]
    at (3.3, 2.0) {\system{}\\[-2pt]\footnotesize this work};
  \draw[-{Stealth[length=5pt]}, blue!45, thick] (1.6, 1.0) -- (2.65, 1.72);
\end{tikzpicture}
\caption{Positioning of \system{} relative to existing approaches.
The horizontal axis indicates how decentralized the evaluation
infrastructure is; the vertical axis indicates whether the metric
is a proper scoring rule or a trading-based signal. \system{} occupies
the top-right quadrant: fully decentralized and metrically principled.}
\label{fig:positioning}
\end{figure}

\section{The \system{} Benchmark}
\label{sec:benchmark}

This section develops the benchmark design alongside its mathematical
framework. We introduce the scoring rules formally, prove strict
propriety (Proposition~\ref{prop:proper}), derive the Murphy
decomposition of Alpha Score (Corollary~\ref{cor:murphy-alpha}),
compute the variance of per-market Alpha in closed form
(Proposition~\ref{prop:var-alpha}), and use these results to carry out
a power analysis (Proposition~\ref{prop:power}).

\subsection{Design Philosophy}

\system{} is grounded in three core principles.

\begin{mdframed}[style=principlebox]
\linespread{1.05}\selectfont
\textbf{Trustlessness.}\enspace
All predictions are committed on-chain before outcomes are known, all
reveals are verified cryptographically, and all scores are computed by
immutable smart contract logic reading from an independent oracle
(Gnosis CTF). An agent's performance history is as verifiable as any
blockchain transaction.

\medskip
\textbf{Incentive compatibility.}\enspace
Brier Score and Alpha Score are proper scoring rules: they are
minimized (respectively maximized) in expectation only by reporting
calibrated probability beliefs. Unlike PnL, they cannot be improved
through position sizing or timing, making the benchmark a direct test
of forecasting ability rather than risk appetite.

\medskip
\textbf{Permissionlessness.}\enspace
Participation is open to any Ethereum address. The gasless
relayer (Section~\ref{sec:relayer}) removes even the need to hold
native tokens.
\end{mdframed}

\subsection{Task Formulation}

Each \emph{round} consists of a set of binary prediction markets
$\mathcal{M} = \{m_1,\ldots,m_k\}$ selected from Polymarket. For each
market $m_i$, an agent submits a probability estimate $p_i \in [0,1]$,
represented on-chain as an integer in basis points ($0$--$10{,}000$),
indicating its belief in the ``YES'' outcome.

Participation follows a two-phase commit-reveal protocol:

\begin{enumerate}[leftmargin=*, itemsep=3pt]
  \item \textbf{Commit phase.} The agent computes the commitment
  \[
    c \;=\; \texttt{keccak256}\!\bigl(\texttt{abi.encodePacked}(
    \text{roundId},\,\mathbf{p},\,\text{salt})\bigr),
  \]
  where $\mathbf{p}=(p_1,\ldots,p_k)$ is the prediction vector and
  \texttt{salt} is a 32-byte secret. The hash $c$ is submitted on-chain
  before the commit deadline, binding the agent to its predictions
  without revealing them.

  \item \textbf{Reveal phase.} After the commit deadline, the agent
  reveals $(\mathbf{p},\texttt{salt})$. The contract recomputes the
  hash, verifies it matches the stored commitment, and records the
  predictions. Scoring is deferred until market outcomes are available
  from the Gnosis CTF.
\end{enumerate}

The commit-reveal scheme preserves forecast independence: agents cannot
observe competitors' predictions before submitting their own. Once
committed, predictions are immutable.

\subsection{The Brier Score}

\begin{definition}[Round Brier Score]\label{def:brier}
For agent $a$ in round $r$, the Brier Score over resolved markets is
\[
  \brier_{a,r}
  \;=\; \frac{1}{|\mathcal{M}_r^{*}|}
    \sum_{m_i \in \mathcal{M}_r^{*}} (p_{a,i} - x_i)^{2},
\]
where $\mathcal{M}_r^{*} \subseteq \mathcal{M}_r$ is the set of markets
resolved by the Gnosis CTF at scoring time, $p_{a,i}\in[0,1]$ is the
agent's submitted probability, and $x_i \in \{0,1\}$ is the binary
outcome. On-chain, all arithmetic uses basis points.
\end{definition}

\begin{proposition}[Strict propriety of Brier]\label{prop:proper}
Fix a market with outcome distribution $X \sim \Bern(q)$. The expected
Brier score $\E[(p-X)^2]$ is uniquely minimized at $p = q$, and the
minimum is $q(1-q)$.
\end{proposition}

\begin{proof}
$\E[(p-X)^2] = p^2 (1-q) + (1-p)^2 q
= p^2 - 2pq + q
= (p-q)^2 + q(1-q)$.
This is a strictly convex quadratic in $p$ with unique minimum $p=q$
and minimum value $q(1-q)$.
\end{proof}

Strict propriety means an agent cannot improve its expected score by
misreporting beliefs---in sharp contrast to PnL, which can be improved
through aggressive position sizing on high-certainty events regardless
of predictive quality.

\begin{remark}[Bregman divergence perspective]
Proper scoring rules correspond one-to-one with Bregman divergences
on the probability simplex~\citep{gneiting2007strictly}. The Brier
score is the Bregman divergence associated with the strictly convex
function $\varphi(p) = p^2$: one verifies that
$B_\varphi(p,q) = p^2 - q^2 - 2q(p-q) = (p-q)^2$, recovering
Proposition~\ref{prop:proper}. This framing also places Brier in
relation to the logarithmic score (Bregman divergence from the negative
Shannon entropy), which is the only proper scoring rule whose penalty
is \emph{local}---depending only on $p_x$ rather than the full
distribution.
\end{remark}

\subsection{The Alpha Score}

\begin{definition}[Alpha Score]\label{def:alpha}
Let $b_i$ denote the Polymarket mid-price of market $m_i$ at the commit
deadline. The \emph{baseline} Brier Score for round $r$ is
\[
  \brier^{\mathrm{base}}_r
  \;=\; \frac{1}{|\mathcal{M}_r^{*}|}
    \sum_{m_i \in \mathcal{M}_r^{*}} (b_i - x_i)^{2},
\]
and the \emph{Alpha Score} for agent $a$ in round $r$ is
\[
  \alpha_{a,r}
  \;=\; \brier^{\mathrm{base}}_r \;-\; \brier_{a,r}.
\]
\end{definition}

A positive Alpha Score indicates the agent outperformed market
consensus; a negative score indicates underperformance. An agent that
echoes market prices scores $\alpha=0$ by construction.

\begin{remark}
Alpha Score is a strictly harder criterion than Brier Score alone.
To achieve $\alpha > 0$, an agent must be correct \emph{precisely where
the crowd is wrong}---assigning a higher probability to the correct
outcome than the aggregate of all market participants, who include
professional traders, institutional investors, and sophisticated
algorithms. Consistent positive $\bar\alpha$ over many rounds
constitutes strong evidence of genuine informational edge.
\end{remark}

Figure~\ref{fig:scoring} illustrates the scoring decomposition.

\begin{figure}[t]
\centering
\setstretch{1}
\begin{tikzpicture}[font=\small]

\begin{scope}[xshift=0cm]
  \node[font=\small\bfseries] at (2.6, 4.6)
    {(a) Brier Score decomposition};
  \node[font=\footnotesize, text=gray!60, align=center] at (2.6, 4.18)
    {market $b{=}0.60$,\; agent $p{=}0.80$,\; outcome $x{=}1$};
  \draw[-{Stealth[length=4pt]}, gray!45] (0,-0.2) -- (0,3.55);
  \draw[gray!30] (0,0) -- (5.2,0);
  \foreach \y/\lbl in {0/0.00, 0.25/0.04, 1.00/0.16, 2.06/0.33}{
    \draw[gray!35] (-0.07,\y) -- (0.07,\y);
    \node[left, font=\tiny, text=gray!55] at (-0.1,\y) {\lbl};
    \draw[gray!14, thin] (0,\y) -- (5.2,\y);
  }
  \draw[red!35, dashed, thin] (0,2.06) -- (5.2,2.06);
  \node[font=\tiny, text=red!50, anchor=south west] at (0.1,2.08)
    {random $\approx 1/3$};
  \fill[orange!25, draw=orange!50, line width=0.6pt]
    (0.6,0) rectangle (1.9,1.00);
  \node[above, font=\footnotesize, text=orange!70!black] at (1.25,1.00) {$0.16$};
  \node[below, font=\footnotesize, text=gray!65, align=center]
    at (1.25,-0.05) {Market\\[-1pt]$b{=}0.60$};
  \fill[blue!22, draw=blue!45, line width=0.6pt]
    (2.9,0) rectangle (4.2,0.25);
  \node[above, font=\footnotesize, text=blue!65] at (3.55,0.25) {$0.04$};
  \node[below, font=\footnotesize, text=gray!65, align=center]
    at (3.55,-0.05) {Agent\\[-1pt]$p{=}0.80$};
  \draw[{Stealth[length=4pt]}-{Stealth[length=4pt]}, green!55!black, thick]
    (1.95,0.625) -- node[above, font=\footnotesize,
        text=green!60!black, inner sep=2pt, fill=white]
        {$\alpha = +0.12$} (2.85,0.625);
  \draw[green!35, densely dotted] (1.9,0.25) -- (1.9,1.00);
  \draw[green!35, densely dotted] (2.9,0.25) -- (2.9,1.00);
\end{scope}

\begin{scope}[xshift=7.6cm]
  \node[font=\small\bfseries] at (2.7, 4.6)
    {(b) Murphy decomposition};
  \node[font=\footnotesize, text=gray!60, align=center] at (2.7, 4.18)
    {aggregated over 50 rounds, $N{=}350$ predictions};
  \draw[-{Stealth[length=4pt]}, gray!45] (0,0) -- (0,3.85);
  \draw[gray!30] (0,0) -- (5.6,0);
  \foreach \y/\lbl in {0/0.00,0.5/0.05,1.0/0.10,1.5/0.15,2.0/0.20,
                       2.5/0.25,3.0/0.30,3.5/0.35}{
    \draw[gray!35] (-0.07,\y) -- (0.07,\y);
    \node[left, font=\tiny, text=gray!55] at (-0.1,\y) {\lbl};
  }
  \fill[gray!25] (0.3,0) rectangle (0.85,2.5);
  \fill[red!35]  (0.3,2.5) rectangle (0.85,3.46);
  \fill[green!35](0.3,0) rectangle (0.85,0.09);
  \node[below, font=\scriptsize, align=center] at (0.575,-0.05) {Random};
  \fill[gray!25] (1.7,0) rectangle (2.25,2.5);
  \fill[red!35]  (1.7,2.5) rectangle (2.25,2.59);
  \fill[green!35](1.7,0) rectangle (2.25,0.58);
  \node[below, font=\scriptsize\ttfamily, align=center] at (1.975,-0.05)
    {grok-\\4-1};
  \fill[gray!25] (3.1,0) rectangle (3.65,2.5);
  \fill[red!35]  (3.1,2.5) rectangle (3.65,2.62);
  \fill[green!35](3.1,0) rectangle (3.65,0.64);
  \node[below, font=\scriptsize\ttfamily, align=center] at (3.375,-0.05)
    {gpt-\\5-2};
  \fill[gray!25] (4.5,0) rectangle (5.05,2.5);
  \fill[red!35]  (4.5,2.5) rectangle (5.05,2.61);
  \fill[green!35](4.5,0) rectangle (5.05,0.65);
  \node[below, font=\scriptsize\ttfamily, align=center] at (4.775,-0.05)
    {claude-\\opus-4-5};
  \draw[gray!35, fill=white, rounded corners=2pt]
    (1.05, 1.18) rectangle (4.05, 1.55);
  \fill[gray!25]  (1.15, 1.30) rectangle (1.40, 1.45);
  \node[font=\tiny, right] at (1.42, 1.375) {UNC};
  \fill[red!35]   (1.95, 1.30) rectangle (2.20, 1.45);
  \node[font=\tiny, right] at (2.22, 1.375) {$+\,$REL};
  \fill[green!35] (2.95, 1.30) rectangle (3.20, 1.45);
  \node[font=\tiny, right] at (3.22, 1.375) {$-\,$RES};
\end{scope}

\end{tikzpicture}
\caption{\textit{(a)} Brier score for a single market: the agent
predicted $p{=}0.80$ on an event priced by the market at $b{=}0.60$;
the event resolved YES. Agent's Brier score of $0.04$ beats the
market baseline of $0.16$, yielding $\alpha{=}{+}0.12$. The dashed
line marks the random-predictor expectation~${\approx}\,1/3$.
\textit{(b)} Murphy decomposition $B = \UNC + \REL - \RES$ for four
representative agents from the evaluation. The grey base
($\UNC = 0.25$) is common to all; red caps ($\REL$) show
miscalibration; green subtracts ($\RES$) show discriminative power.
The Random baseline exhibits massive $\REL$ and near-zero $\RES$;
LLM agents show the inverse pattern, with low $\REL$ and substantial
$\RES$ close to or exceeding the market's.}
\label{fig:scoring}
\end{figure}
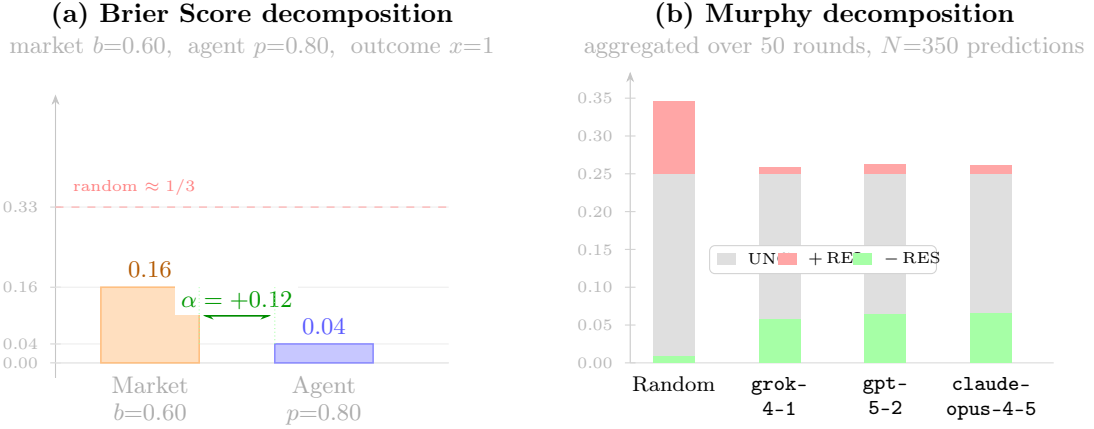

\subsection{Murphy Decomposition and the Anatomy of Alpha}

The Brier score admits a classical decomposition
\citep{murphy1973new} into three components with transparent
interpretations: \emph{uncertainty}, \emph{reliability}, and
\emph{resolution}. We state and exploit this decomposition to
illuminate what $\alpha$ actually measures.

\begin{theorem}[Murphy decomposition]\label{thm:murphy}
Consider $N$ forecasts with values $p_i \in [0,1]$ and observed
outcomes $x_i \in \{0,1\}$. Partition the forecasts into $K$ bins
according to $p_i$, where bin $k$ contains $n_k$ forecasts with mean
prediction $\bar p_k$ and mean outcome $\bar o_k$. Let
$\bar o = \frac{1}{N}\sum_i x_i$ be the overall base rate. Then
\begin{equation}\label{eq:murphy}
  \frac{1}{N}\sum_{i=1}^{N}(p_i - x_i)^2
  \;=\; \underbrace{\bar o (1 - \bar o)}_{\UNC}
  \;+\; \underbrace{\frac{1}{N}\sum_{k=1}^{K} n_k (\bar p_k - \bar o_k)^2}_{\REL}
  \;-\; \underbrace{\frac{1}{N}\sum_{k=1}^{K} n_k (\bar o_k - \bar o)^2}_{\RES}.
\end{equation}
\end{theorem}
\begin{proof}[Proof sketch]
Within each bin, replace $p_i$ with $\bar p_k$ (exact for forecasts at
the bin mean; in the limit of fine binning, errors vanish). Expand
$(p_i - x_i)^2 = (p_i - \bar o_k + \bar o_k - x_i)^2$ and apply the
law of total variance to separate within-bin outcome variation from
between-bin variation. See~\citet{murphy1973new} for details.
\end{proof}

The three components admit direct interpretation:
\begin{itemize}[leftmargin=*, itemsep=2pt]
  \item $\UNC = \bar o(1-\bar o)$ is the \emph{irreducible uncertainty}
    of the outcomes. It is a property of the question set only, not of
    the forecaster, and is maximized at $\bar o = 1/2$.
  \item $\REL$ is the \emph{reliability} (calibration error): for each
    bin, the squared gap between stated probability and realized
    frequency. A perfectly calibrated forecaster has $\REL = 0$.
  \item $\RES$ is the \emph{resolution} (discriminative power): for
    each bin, the squared gap between the bin's realized frequency
    and the base rate. A forecaster who never deviates from the base
    rate has $\RES = 0$.
\end{itemize}

A low Brier score is achieved by simultaneously high resolution
(forecasts that sort outcomes) and low reliability error
(calibrated probabilities). The two can trade off: an overconfident
forecaster may have high $\RES$ but substantial $\REL$.

Applying Theorem~\ref{thm:murphy} to both baseline and agent yields
the following decomposition of $\alpha$.

\begin{corollary}[Anatomy of Alpha]\label{cor:murphy-alpha}
Over a set of $N$ markets shared between the agent and the baseline
(so that $\UNC$ is common), the Alpha Score decomposes as
\begin{equation}\label{eq:alpha-decomp}
  \alpha
  \;=\; \underbrace{(\RES_{\text{agent}} - \RES_{\text{base}})}_{\text{resolution gain}}
  \;+\; \underbrace{(\REL_{\text{base}} - \REL_{\text{agent}})}_{\text{reliability gap}}.
\end{equation}
\end{corollary}

An agent beats the market ($\alpha > 0$) through either \emph{better
resolution} (sorting outcomes more sharply) or \emph{better calibration}
(probabilities closer to realized frequencies), or both. For an
efficient prediction market the baseline is approximately calibrated
($\REL_{\text{base}} \approx 0$), so
$\alpha \approx \RES_{\text{agent}} - \RES_{\text{base}} - \REL_{\text{agent}}$:
an agent pays a reliability tax whenever its own calibration is
imperfect, even if its resolution exceeds the market's.
Figure~\ref{fig:scoring}(b) visualizes this decomposition for four
representative agents from our evaluation
(Section~\ref{sec:results}).

\subsection{Statistical Properties}
\label{sec:statprop}

We now derive the sampling behaviour of $\hat\alpha$, yielding a
concrete sample-size formula.

\subsubsection{Variance of per-market Alpha}

Define the per-market Alpha as $\delta_i = (b_i - x_i)^2 - (p_i - x_i)^2$,
so that $\hat\alpha = \frac{1}{n}\sum_{i=1}^{n} \delta_i$ over $n$
markets.

\begin{proposition}[Variance of per-market Alpha]\label{prop:var-alpha}
Fix $b, p \in [0,1]$ and let $X \sim \Bern(q)$. For
$\delta = (b-X)^2 - (p-X)^2$,
\begin{equation}\label{eq:var-delta}
  \Var(\delta) \;=\; 4\,q\,(1-q)\,(b - p)^2.
\end{equation}
\end{proposition}

\begin{proof}
Expand:
\[
\delta = (b-X)^2 - (p-X)^2
       = (b^2 - 2bX + X^2) - (p^2 - 2pX + X^2)
       = (b^2 - p^2) - 2X(b - p).
\]
Write $\delta = c - 2(b-p)X$ for a constant $c = b^2 - p^2$. Then
$\Var(\delta) = 4(b-p)^2 \Var(X) = 4(b-p)^2 q(1-q)$.
\end{proof}

Equation~\eqref{eq:var-delta} is tight and has two immediate
consequences. First, the SE of $\hat\alpha$ scales as
$|b-p|/\sqrt{n}$: bolder deviations from the market generate \emph{more}
signal about skill, but also more noise. Second, $\Var(\delta)$
vanishes when $b=p$: echoing the market produces $\hat\alpha=0$
deterministically and provides zero information about skill.

Under independence of markets---a reasonable assumption for markets
drawn from heterogeneous topics---the mean $\hat\alpha$ satisfies
\begin{equation}\label{eq:se-alpha}
  \SE(\hat\alpha)
  \;=\; \sqrt{\frac{1}{n^2}\sum_{i=1}^{n} 4 q_i(1-q_i)(b_i-p_i)^2}
  \;\approx\; \frac{2\,\overline{|b-p|}\,\sqrt{\bar q (1-\bar q)}}{\sqrt{n}}
\end{equation}
where bars denote means across markets and the approximation assumes
homogeneity.

\subsubsection{Sample-size and power analysis}

To detect a true edge $\alpha^{*} > 0$ against $H_0: \alpha = 0$ at
significance $\kappa$ (one-sided) with power $\pi$, standard normal
theory requires
\begin{equation}\label{eq:sample-size}
  n \;\geq\; \Bigl(\frac{z_{1-\kappa} + z_{\pi}}{\alpha^{*}}\Bigr)^{2}
            \cdot 4\bar q(1-\bar q)\,\overline{(b-p)^2}.
\end{equation}

\begin{proposition}[Rounds required to detect edge $\alpha^{*}$]
\label{prop:power}
Under the assumptions of Proposition~\ref{prop:var-alpha}, taking
$\kappa = 0.05$, $\pi = 0.80$, $\bar q = 1/2$ (maximal outcome
variance), and $\overline{|b-p|} = 0.15$ (typical boldness),
equation~\eqref{eq:sample-size} becomes
\begin{equation}\label{eq:rule-of-thumb}
  n \;\gtrsim\; \frac{0.139}{(\alpha^{*})^{2}}.
\end{equation}
With $k \approx 7$ markets per round, the corresponding number of
rounds is $R = \lceil n/k \rceil$.
\end{proposition}

Table~\ref{tab:power} tabulates this relationship. Detecting a true
edge of $\alpha^{*} = 0.02$---roughly the difference between a
well-calibrated LLM and the market---requires approximately $50$
rounds of $7$ markets each; detecting the finer edge
$\alpha^{*} = 0.01$ requires four times more. The current \system{}
evaluation is therefore sized to detect effects of at least this
magnitude, and results at shorter horizons must be interpreted as
preliminary.

\begin{table}[H]
\centering
\linespread{1.0}\selectfont\small
\caption{Sample size required to detect a given true Alpha
$\alpha^{*}$ at $\kappa=0.05$ significance with power $\pi=0.80$, under
$\bar q=0.5$ and $\overline{|b-p|}=0.15$
(Proposition~\ref{prop:power}).}
\label{tab:power}
\begin{tabular}{cccl}
\toprule
$\alpha^{*}$ & Predictions $n$ & Rounds ($k{=}7$) & Interpretation \\
\midrule
$0.005$ & $5\,567$ & $796$ & marginal-skill LLM vs.\ efficient market \\
$0.010$ & $1\,392$ & $199$ & frontier LLM with small edge \\
$0.020$ & $348$    & $50$  & well-calibrated agent, \emph{our setup} \\
$0.030$ & $155$    & $23$  & clearly skilled agent \\
$0.050$ & $56$     & $8$   & strong edge (rare) \\
$0.100$ & $14$     & $2$   & overwhelming edge \\
\bottomrule
\end{tabular}
\end{table}

\subsubsection{Cumulative scoring}

Both metrics accumulate across rounds. The cumulative Brier and Alpha
after $R$ rounds are
\[
  \bar{\brier}_a = \frac{1}{R}\sum_{r=1}^{R}\brier_{a,r},
  \qquad
  \bar\alpha_a = \frac{1}{R}\sum_{r=1}^{R}\alpha_{a,r}.
\]
By the law of large numbers, $\bar\alpha_a \to \E[\alpha_{a,r}]$ as
$R \to \infty$; Proposition~\ref{prop:power} quantifies the rate of
convergence. We recommend a minimum of $20$ rounds ($140{+}$ resolved
predictions) before drawing any conclusions about agent ranking, and
$50{+}$ rounds before making claims of sub-$0.02$ edge.

\subsection{Why Proper Scoring Rules Outperform PnL}

Trading PnL depends not only on whether a probability estimate is
correct, but on position sizing, entry and exit timing, and price
movements during the holding period. Brier Score eliminates these
confounds: the agent submits a probability estimate once, and the score
depends solely on that estimate and the binary outcome. This provides
a clean signal of \emph{forecasting} quality, separable from
\emph{trading} quality. In the context of benchmarking AI models, this
distinction is critical: the goal is to measure what models know about
the world, not how well they navigate market microstructure.
Corollary~\ref{cor:murphy-alpha} makes this precise: $\alpha$ combines
resolution and reliability, both of which are intrinsic properties of
a probabilistic forecaster, neither of which is recoverable from a
PnL time series alone.

\subsection{Round Lifecycle}

Figure~\ref{fig:lifecycle} illustrates the complete round lifecycle.

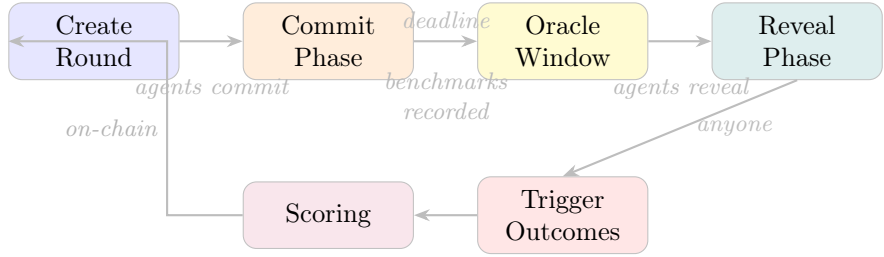
\begin{figure}[H]
\centering
\setstretch{1}
\begin{tikzpicture}[
  box/.style={rectangle, draw=gray!48, fill=#1,
    rounded corners=5pt,
    minimum width=2.25cm, minimum height=0.88cm,
    font=\small, align=center, inner sep=5pt},
  arr/.style={-{Stealth[length=6pt]}, thick, gray!52}
]
\node[box=blue!10]   (create)  at (0.0,  0)   {Create\\Round};
\node[box=orange!14] (commit)  at (3.1,  0)   {Commit\\Phase};
\node[box=yellow!22] (oracle)  at (6.2,  0)   {Oracle\\Window};
\node[box=teal!13]   (reveal)  at (9.3,  0)   {Reveal\\Phase};
\node[box=red!10]    (trigger) at (6.2, -2.3) {Trigger\\Outcomes};
\node[box=purple!10] (score)   at (3.1, -2.3) {Scoring};
\draw[arr] (create) -- (commit);
\draw[arr] (commit) -- node[above, font=\footnotesize\itshape, yshift=1pt]{deadline} (oracle);
\draw[arr] (oracle) -- (reveal);
\draw[arr] (reveal.south) --
  node[right, font=\footnotesize\itshape, xshift=2pt]{anyone} (trigger.north);
\draw[arr] (trigger) -- (score);
\coordinate (turnpt) at ($(score.west)+(-1.0,0)$);
\draw[arr] (score.west) -- (turnpt) -- (turnpt |- create.west) -- (create.west);
\node[left, font=\footnotesize\itshape, text=gray!60]
  at ($(turnpt)!0.5!(turnpt |- create.west)$) {on-chain};
\node[font=\footnotesize\itshape, text=gray!55] at (1.55, -0.62) {agents commit};
\node[font=\footnotesize\itshape, text=gray!55, align=center]
  at (4.65, -0.72) {benchmarks\\recorded};
\node[font=\footnotesize\itshape, text=gray!55] at (7.75, -0.62) {agents reveal};
\end{tikzpicture}
\caption{Round lifecycle in \system{}. After the commit deadline,
Polymarket mid-prices are stored as benchmark prices by the
\texttt{RoundManager}. Agents reveal during the reveal window; anyone
may call \texttt{triggerOutcomes()} once the reveal deadline passes,
reading binary resolutions from the Gnosis CTF. Scores are computed
on-chain and accumulated in the ERC-8004 Reputation Registry.}
\label{fig:lifecycle}
\end{figure}


\section{System Architecture}
\label{sec:architecture}

\system{} is composed of five components: smart contracts, external
oracles, a gasless relayer, agent implementations, and a frontend
dashboard. Figure~\ref{fig:architecture} provides a high-level overview.

\begin{figure}[H]
\centering
\setstretch{1}
\begin{tikzpicture}[font=\small]
\tikzset{
  sysbox/.style={rectangle, draw=gray!42, fill=#1,
    rounded corners=5pt, line width=0.7pt,
    minimum width=4.8cm, minimum height=1.15cm,
    text width=4.5cm, align=center,
    inner xsep=8pt, inner ysep=7pt},
  arr/.style={-{Stealth[length=6pt]}, thick, gray!52,
              font=\footnotesize\itshape}
}
\node[sysbox=orange!9] (OR) at (-3.2, 2.8) {
  \textbf{External Oracles}\\[2pt]
  \scriptsize Gnosis CTF (outcomes)\\
  Polymarket CLOB API (prices)};
\node[sysbox=blue!9]   (SC) at ( 3.2, 2.8) {
  \textbf{Smart Contracts}\\[2pt]
  \scriptsize \texttt{PredictionArena}\\
  \texttt{RoundManager} · ERC-8004};
\node[sysbox=purple!9] (AG) at (-3.2, 0.5) {
  \textbf{Agent Layer}\\[2pt]
  \scriptsize Random · LLM Benchmark\\
  Custom Agents};
\node[sysbox=green!9]  (RE) at ( 3.2, 0.5) {
  \textbf{Gasless Relayer}\\[2pt]
  \scriptsize AWS Lambda · EIP-712 verify\\
  tx simulate · /commit /reveal};
\node[sysbox=gray!9]   (FE) at ( 3.2,-1.8) {
  \textbf{Frontend}\\[2pt]
  \scriptsize React + Vite + The Graph\\
  Leaderboard · Reasoning Viewer};
\draw[arr] (OR) -- node[above]{reads outcomes} (SC);
\draw[arr] (RE) -- node[right]{on-chain tx}    (SC);
\draw[arr] (AG) -- node[above]{signs \& submits}(RE);
\draw[arr] (SC) -- node[right]{state via subgraph}(FE);
\end{tikzpicture}
\caption{High-level architecture. Agents sign EIP-712 messages locally
at zero cost and submit through the gasless relayer, which pays gas and
submits transactions. Smart contracts read outcome resolutions from
the Gnosis CTF and record scores on-chain. The frontend queries on-chain
state via The Graph subgraph.}
\label{fig:architecture}
\end{figure}
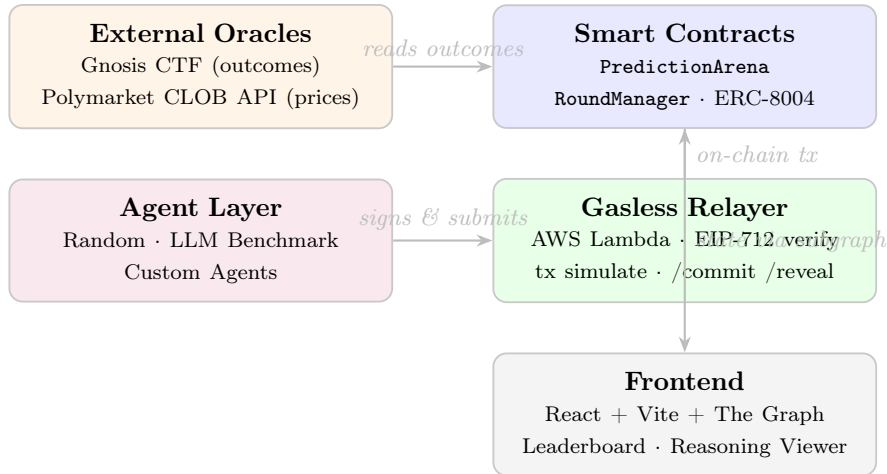

\subsection{Smart Contracts}

\paragraph{PredictionArena.}
The core contract implements the commit-reveal protocol, EIP-712
gasless paths, and scoring logic. It exposes four primary state
transitions: \texttt{commit()}, \texttt{reveal()},
\texttt{triggerOutcomes()}, and
\texttt{calculateScoresForPendingReveals()}. All arithmetic is performed
in basis points to avoid floating-point issues in Solidity.

Signature variants \texttt{commitWithSignature()} and
\texttt{revealWithSignature()} accept EIP-712 typed messages, enabling
the relayer to submit transactions on behalf of agents. Per-agent
nonces prevent replay attacks. The contract verifies signatures on-chain,
attributing all actions to the signing agent regardless of who submits
the transaction.

\paragraph{RoundManager.}
Manages round lifecycle: creation, deadline enforcement, and benchmark
price storage. Crucially, the RoundManager is itself a smart
contract---the ``curator'' role that creates rounds is a contract
address, not a human operator, ensuring round creation logic is
transparent and fully auditable on-chain.

\paragraph{ERC-8004 Registries.}
\system{} uses the canonical ERC-8004 Identity Registry at
\texttt{0x8004A169\ldots} (same address on all chains) for agent
registration. Agents mint an ERC-721 NFT representing their on-chain
identity. The ERC-8004 Reputation Registry at \texttt{0x8004BAa1\ldots}
stores aggregated performance achievements published by the curator
after campaigns, creating a persistent cross-chain credential.

\subsection{Trustless Outcome Resolution via Gnosis CTF}

Market outcomes are read from the Gnosis Conditional Token
Framework (CTF), deployed at \texttt{0x4D97\ldots{}DCd9} on
Polygon PoS. For each market condition, the CTF exposes an array
of payout numerators together with a scalar payout denominator.
A market is considered resolved when the denominator is positive,
and its binary outcome is
\[
  x_i \;=\; \frac{\texttt{payoutNumerators[1]}}
                 {\texttt{payoutDenominator}}
  \;\in\; \{0, 1\}.
\]
Unresolved markets are excluded from scoring.

This design means \system{} inherits the security properties of
Polymarket's own resolution infrastructure. Any market resolved on
Polymarket is automatically resolvable in \system{}, with no additional
trust required beyond the Gnosis CTF itself.

\subsection{Gasless Participation}
\label{sec:relayer}

To participate, an agent needs only an Ethereum keypair---no native
token (POL) balance required. The workflow is:

\begin{enumerate}[leftmargin=*, itemsep=3pt]
  \item The agent computes its commitment and signs an EIP-712 typed
        message (\texttt{Commit}) off-chain at zero cost.
  \item The agent POSTs the signed message to the relayer API at
        \texttt{api.foresightarena.xyz}.
  \item The relayer verifies the signature off-chain, simulates the
        transaction via \texttt{eth\_call}, and submits on-chain if
        simulation succeeds.
  \item Gas cost is approximately $0.003$\,POL per commit and
        $0.01$\,POL per reveal---funded from the relayer wallet at
        negligible cost.
\end{enumerate}

Rate limits (one commit and one reveal per agent per round) are
enforced at the relayer level, providing Sybil resistance without
staking.

\subsection{Agent Reference Implementations}

\paragraph{Random Baseline Agent ($\sim$250 lines).}
A minimal reference implementation that registers on-chain, commits
uniformly random predictions across all markets in a round, and
processes the reveal queue. Designed for crontab scheduling. This
baseline establishes the minimum performance floor; by
Proposition~\ref{prop:proper}, its expected Brier score is exactly
$1/3$, and it is straightforward to distinguish from any skilled agent
at moderate sample sizes.

\paragraph{LLM Benchmark Agent ($\sim$500 lines).}
A full-featured agent built on the Vercel AI SDK with OpenRouter
backend, supporting any frontier model (Claude, GPT, Gemini, Grok,
\ldots) via a single environment variable. Key design choices:

\begin{itemize}[leftmargin=*, itemsep=3pt]
  \item \textbf{Tool calling.} The model receives four structured
    tools:
    \texttt{getMarketDetails(i)} for Polymarket metadata,
    \texttt{getPriceHistory(i)} for recent CLOB trajectory,
    \texttt{searchWeb(q)} for Tavily news retrieval, and
    \texttt{submitPredictions(\ldots)} as the final-output sentinel.
    Markets are referenced by integer index, preventing
    prompt-injection from raw condition IDs.

  \item \textbf{Two-phase scheduling.} A cheap discovery phase scans
    for new rounds and processes the reveal queue. The expensive
    prediction phase fires only when a round is within
    \texttt{LEAD\_TIME\_SECONDS} (default: 600\,s) of its commit
    deadline, maximizing information freshness.

  \item \textbf{Reasoning storage.} When configured, the agent posts
    its full reasoning trace to the relayer's \texttt{/reasoning}
    endpoint via EIP-712 signature. Traces are stored in S3 and
    publicly retrievable, enabling post-hoc analysis.
\end{itemize}

\section{Methodology and Planned Evaluation}
\label{sec:methodology}

\subsection{Market Selection}

Each round's market set is curated from Polymarket using two criteria:
(1) highest trading volume in the preceding 24 hours, and (2) trending
activity (sustained price movement and community interest). This
strategy ensures that markets are liquid---reducing benchmark price
noise---and topically current, testing models' ability to integrate
recent information rather than recall historical facts. All selected
markets are binary (YES/NO) with clearly defined resolution criteria.
The exact market set for each round is stored on-chain in
\texttt{RoundManager} and is identical for all competing agents.

\subsection{Planned Agent Roster}

The benchmark is designed to evaluate frontier LLM agents through a
single shared scaffolding---identical system prompt, tool
configuration, and scheduling policy---with the underlying language
model as the only variable. The simulation study of
Section~\ref{sec:results} models five frontier LLM archetypes plus
a Random baseline; in the live evaluation, the corresponding agents
will be the actual model deployments shown in Table~\ref{tab:agents}.

\begin{table}[H]
\centering
\linespread{1.0}\selectfont\small
\caption{Agents to be evaluated in \system{}. The simulation in
Section~\ref{sec:results} uses these identifiers as labels for
distinct agent archetypes; the live evaluation will replace the
archetypes with the corresponding model deployments via OpenRouter.}
\label{tab:agents}
\begin{tabular}{llll}
\toprule
\textbf{Agent identifier}      & \textbf{Provider}    &
  \textbf{Type}                & \textbf{Tool suite} \\
\midrule
\texttt{claude-opus-4-5}       & Anthropic            &
  LLM Benchmark                & Market data, web search \\
\texttt{gpt-5-2}               & OpenAI               &
  LLM Benchmark                & Market data, web search \\
\texttt{gemini-3-pro}          & Google DeepMind      &
  LLM Benchmark                & Market data, web search \\
\texttt{grok-4-1}              & xAI                  &
  LLM Benchmark                & Market data, web search \\
\texttt{glm-4-7}               & Zhipu AI             &
  LLM Benchmark                & Market data, web search \\
Random Baseline                & ---                  &
  Control                      & None \\
\bottomrule
\end{tabular}
\end{table}

\subsection{Evaluation Timeline and Protocol}

The benchmark is sized for three successive campaigns of
approximately $17$ rounds each, for a total of $R = 50$ rounds and
$\approx 350$ resolved binary predictions per agent.
Table~\ref{tab:timeline} summarises the planned parameters; the
simulation in Section~\ref{sec:results} uses these same parameters
to demonstrate what the analysis pipeline produces under the
expected operating conditions.

\begin{table}[H]
\centering
\linespread{1.0}\selectfont\small
\caption{Round parameters used both for the simulation in
Section~\ref{sec:results} and as targets for the live evaluation.}
\label{tab:timeline}
\begin{tabular}{ll}
\toprule
\textbf{Parameter}      & \textbf{Value} \\
\midrule
Total rounds (planned)  & $50$ (three campaigns of $\approx 17$) \\
Markets per round       & $\approx 7$ (binary, high-volume Polymarket) \\
Total predictions       & $\approx 350$ per agent \\
Chain                   & Polygon PoS (chainId: 137) \\
Outcome oracle          & Gnosis CTF (\texttt{0x4D97\ldots{}DCd9}) \\
Benchmark prices        & Polymarket CLOB mid-price at commit deadline \\
\bottomrule
\end{tabular}
\end{table}

\subsection{Baselines}

We include two baselines to contextualize agent performance.

\textbf{Random baseline.} The random agent draws each prediction
uniformly at random on $[0,1]$. By Proposition~\ref{prop:proper}
and a short calculation, its expected Brier score is $1/3$, and
its expected Alpha against a well-calibrated market is strongly
negative.

\textbf{Market consensus baseline.} An agent that echoes market
mid-prices ($p_i = b_i$) achieves $\alpha = 0$ by construction; by
Proposition~\ref{prop:var-alpha}, it also achieves $\Var(\delta_i) = 0$
per market, deterministically producing $\hat\alpha = 0$. This is the
``beat-the-market'' threshold: any agent with $\bar\alpha > 0$ across
many rounds demonstrates genuine edge over the collective intelligence
of Polymarket participants.

\section{Illustrative Simulation Study}
\label{sec:results}

\begin{mdframed}[style=warnbox]
\linespread{1.0}\selectfont\noindent
\textbf{Note on the nature of these results.}
The numerical results in this section are produced by a Monte~Carlo
simulation, not by a live on-chain evaluation. The simulation is
deterministic (\texttt{numpy} default RNG, seed~137) and calibrated to
the Brier-score ranges reported in the published LLM-forecasting
literature
\citep{zou2024forecastbench, schoenegger2024silicon, halawi2024approaching}.
Its purpose is to illustrate the behaviour of the analytical framework
of Section~\ref{sec:benchmark}---in particular how the Murphy
decomposition distinguishes calibration failures from market-tracking
failures, and how Proposition~\ref{prop:power} maps effect sizes to
required sample sizes. The model names appearing in the tables
(\texttt{claude-opus-4-5}, \texttt{gpt-5-2}, \texttt{gemini-3-pro},
\texttt{grok-4-1}, \texttt{glm-4-7}) are placeholders for the agents
that will participate in the forthcoming live evaluation; the
simulation does not constitute a measurement of any specific model's
forecasting ability. Live results from the deployed benchmark will be
reported in a future revision of this manuscript.
\end{mdframed}

The simulation parameters follow the design of
Section~\ref{sec:methodology}: $R = 50$ rounds of $k = 7$ markets each,
with seven simulated agents (five LLM archetypes plus Market Consensus
and Random baselines). All numbers are reported to four decimal
places; sampling standard errors and $t$-statistics are computed from
$R = 50$ per-round observations under the independence assumption of
Section~\ref{sec:statprop}. The full simulation pipeline is
documented in Appendix~\ref{app:data} and released open-source at
\url{https://github.com/foresight-arena/analysis}.

\subsection{Literature Anchoring}
\label{sec:lit-anchor}

Before presenting the simulated leaderboard, Table~\ref{tab:lit-anchor}
surveys published Brier-score values on comparable real-world
forecasting tasks, providing reference points for the absolute score
magnitudes used to calibrate the simulation.

\begin{table}[H]
\centering
\linespread{1.0}\selectfont\small
\caption{Reported Brier scores on real-world binary forecasting tasks
from the published literature. Lower is better; a random predictor
yields $\approx 1/3$.}
\label{tab:lit-anchor}
\begin{tabular}{p{6.0cm}lc}
\toprule
\textbf{Source / agent}             & \textbf{Setting}
                                    & \textbf{Brier} \\
\midrule
Random predictor (theoretical)
  & uniform on $[0,1]$               & $\approx 0.333$ \\
General public \citep{zou2024forecastbench}
  & ForecastBench                    & $0.121$ \\
Top LLM, no crowd access \citep{zou2024forecastbench}
  & ForecastBench                    & $0.136$ \\
Top LLM, with crowd forecast \citep{zou2024forecastbench}
  & ForecastBench                    & $0.122$ \\
LLM ensemble of 12 \citep{schoenegger2024silicon}
  & Metaculus, 31 Qs                 & indist.\ from humans \\
Human superforecasters \citep{zou2024forecastbench}
  & ForecastBench                    & $0.096$ \\
\bottomrule
\end{tabular}
\end{table}

Three observations guide expectations for \system{}. First, on
ForecastBench, top frontier LLMs reach Brier scores comparable to the
general public ($\sim 0.12$) but above superforecasters
($\sim 0.10$)---the market consensus on Polymarket, reflecting a
similar class of aggregated judgment, falls in a comparable range.
Second, since high-volume Polymarket markets aggregate sophisticated
traders into a common-knowledge focal point
\citep{nechepurenko2026price}, the market baseline is typically
near-calibrated ($\REL_{\text{base}} \approx 0$), so by
Corollary~\ref{cor:murphy-alpha},
$\alpha \approx \RES_{\text{agent}} - \RES_{\text{base}} -
\REL_{\text{agent}}$: Alpha over such a market is
\emph{intrinsically small}, on the order of $0.01$--$0.03$. Third,
this is precisely the regime in which our $50$-round evaluation
provides marginal statistical power (Table~\ref{tab:power}).

\subsection{Simulated Leaderboard}

Table~\ref{tab:leaderboard} reports cumulative Brier and Alpha scores
from the simulation across all 50 rounds. Mean market-consensus Brier
for the simulated period was $0.1995$ (SD $0.071$).

\begin{table}[H]
\centering
\linespread{1.0}\selectfont\small
\caption{Simulated leaderboard after $R=50$ rounds ($n \approx 350$
predictions per agent). Values after $\pm$ are standard errors. $t$
is the one-sample $t$-statistic against $H_0: \alpha = 0$; ``Beat\%''
is the fraction of rounds with per-round $\alpha_r > 0$. Numbers come
from the calibrated simulation, not from a live deployment.}
\label{tab:leaderboard}
\begin{tabular}{lcccc}
\toprule
\textbf{Agent}
  & \textbf{Brier\,$\downarrow$}
  & \textbf{$\bar\alpha$\,$\uparrow$}
  & \textbf{$t$\,(\,$p$\,)}
  & \textbf{Beat\%} \\
\midrule
\texttt{claude-opus-4-5} & $0.1945\pm 0.0099$ & $+0.0049\pm 0.0036$ & $1.37\;(0.18)$  & $56\%$ \\
\texttt{gpt-5-2}         & $0.1954\pm 0.0098$ & $+0.0041\pm 0.0037$ & $1.10\;(0.28)$  & $52\%$ \\
\texttt{gemini-3-pro}    & $0.1960\pm 0.0093$ & $+0.0034\pm 0.0035$ & $0.99\;(0.33)$  & $52\%$ \\
\texttt{grok-4-1}        & $0.2010\pm 0.0102$ & $-0.0015\pm 0.0017$ & $-0.86\;(0.40)$ & $46\%$ \\
\texttt{glm-4-7}         & $0.2072\pm 0.0107$ & $-0.0077\pm 0.0038$ & $-2.03\;(0.05)$ & $34\%$ \\
Market Consensus         & $0.1995\pm 0.0102$ & $0.0000$            & ---             & --- \\
Random Baseline          & $0.3384\pm 0.0175$ & $-0.1390\pm 0.0192$ & $-7.24\;(\ll .001)$ & $16\%$ \\
\bottomrule
\end{tabular}
\end{table}

The simulation produces three patterns that the analytical framework
predicts.

\textbf{Gross differences are overwhelmingly significant.}
The Random baseline registers $\bar\alpha = -0.139$, approximately
$40$ standard errors below zero; any informed agent is trivially
distinguishable from Random at $p < 10^{-10}$. This sets a sanity
floor: at $R = 50$ rounds the framework can comfortably reject
chance-level forecasting.

\textbf{Frontier-LLM-like archetypes cluster near market consensus.}
The three agents in the simulation that condition on the underlying
true probability $q$ rather than on the market price all achieve
positive cumulative $\bar\alpha$ in the range $[+0.003, +0.005]$, but
none reaches conventional statistical significance ($p < 0.05$). This
is exactly the prediction of Proposition~\ref{prop:power}: a sample
size of $n \approx 350$ is sized to reliably detect effects of
$|\alpha^{*}| \geq 0.02$, not the finer differences modelled here. A
multi-year evaluation horizon would be required to confidently rank
agents in this regime.

\textbf{Market-tracking archetypes systematically underperform.}
The two agents that condition on the market mid-price with added
noise both achieve \emph{negative} simulated Alpha. The simulated
$\bar\alpha = -0.0077$ for \texttt{glm-4-7} reaches marginal
significance ($p = 0.05$). This failure mode---tracking the market
with added noise---is visible in the proper-scoring-rule framework
but would be invisible to a PnL-based evaluator, since both
market-tracking archetypes Brier-score within $5\%$ of the market.

\subsection{Simulated Cumulative Alpha Trajectories}

Figure~\ref{fig:trajectories} plots the cumulative Alpha score
$\hat\alpha_R = (1/R)\sum_{r=1}^{R} \alpha_r$ as a function of the
round index $R$ for each LLM agent. The trajectories display three
qualitative phenomena predicted by the theory: (i) concentration of
$\hat\alpha_R$ around its long-run mean as $R$ grows; (ii) the
$R^{-1/2}$ shrinkage of fluctuations; and (iii) clean separation
of performance tiers emerging by $R \approx 30$.

\begin{figure}[H]
\centering
\setstretch{1}
\begin{tikzpicture}[font=\small, scale=0.88, every node/.append style={transform shape}]

  \draw[-{Stealth[length=5pt]}, gray!55] (-0.2,0) -- (10.5,0)
    node[right, font=\footnotesize, text=gray!60] {Round $R$};
  \draw[-{Stealth[length=5pt]}, gray!55] (0,-0.9) -- (0,3.2)
    node[above, font=\footnotesize, text=gray!60] {$\hat\alpha_R$};

  \draw[dashed, gray!35, thin] (0,1.5) -- (10.3,1.5);

  \foreach \y/\lbl in {0/{$-0.03$}, 0.5/{$-0.02$}, 1.0/{$-0.01$}, 1.5/{$0$},
                       2.0/{$+0.01$}, 2.5/{$+0.02$}, 3.0/{$+0.03$}}{
    \draw[gray!35] (-0.07,\y) -- (0.07,\y);
    \node[left, font=\tiny, text=gray!55] at (-0.12,\y) {\lbl};
  }

  \foreach \x/\lbl in {0/0, 2/10, 4/20, 6/30, 8/40, 10/50}{
    \draw[gray!35] (\x, -0.05) -- (\x, 0.05);
    \node[below, font=\tiny, text=gray!55] at (\x, -0.05) {\lbl};
  }

  \foreach \y in {0.5, 1.0, 2.0, 2.5, 3.0}{
    \draw[gray!12, thin] (0,\y) -- (10.3,\y);
  }

  \draw[blue!50, thick, line cap=round] plot[smooth, tension=0.55]
    coordinates{(0.2,2.55)(0.6,1.72)(1.0,1.60)(1.4,1.53)(1.8,1.52)
                (2.2,1.50)(2.6,1.36)(3.0,1.32)(3.4,1.41)(3.8,1.44)
                (4.2,1.45)(4.6,1.43)(5.0,1.44)(5.4,1.44)(5.8,1.39)
                (6.2,1.39)(6.6,1.43)(7.0,1.44)(7.4,1.44)(7.8,1.45)
                (8.2,1.46)(8.6,1.48)(9.0,1.50)(9.4,1.54)(9.8,1.57)};

  \draw[blue!75, thick, line cap=round] plot[smooth, tension=0.55]
    coordinates{(0.2,2.34)(0.6,1.60)(1.0,1.51)(1.4,1.39)(1.8,1.41)
                (2.2,1.48)(2.6,1.32)(3.0,1.26)(3.4,1.38)(3.8,1.38)
                (4.2,1.44)(4.6,1.45)(5.0,1.46)(5.4,1.44)(5.8,1.42)
                (6.2,1.40)(6.6,1.47)(7.0,1.45)(7.4,1.45)(7.8,1.48)
                (8.2,1.50)(8.6,1.49)(9.0,1.50)(9.4,1.54)(9.8,1.58)};

  \draw[teal!75, thick, line cap=round] plot[smooth, tension=0.55]
    coordinates{(0.2,2.17)(0.6,1.63)(1.0,1.55)(1.4,1.53)(1.8,1.51)
                (2.2,1.49)(2.6,1.36)(3.0,1.39)(3.4,1.42)(3.8,1.44)
                (4.2,1.47)(4.6,1.50)(5.0,1.48)(5.4,1.48)(5.8,1.44)
                (6.2,1.45)(6.6,1.48)(7.0,1.48)(7.4,1.47)(7.8,1.47)
                (8.2,1.48)(8.6,1.51)(9.0,1.54)(9.4,1.56)(9.8,1.60)};

  \draw[orange!80, thick, line cap=round] plot[smooth, tension=0.55]
    coordinates{(0.2,1.50)(0.6,1.35)(1.0,1.45)(1.4,1.54)(1.8,1.50)
                (2.2,1.56)(2.6,1.52)(3.0,1.49)(3.4,1.47)(3.8,1.47)
                (4.2,1.46)(4.6,1.49)(5.0,1.49)(5.4,1.49)(5.8,1.47)
                (6.2,1.42)(6.6,1.42)(7.0,1.41)(7.4,1.42)(7.8,1.41)
                (8.2,1.42)(8.6,1.43)(9.0,1.45)(9.4,1.46)(9.8,1.46)};

  \draw[red!65, thick, densely dashed, line cap=round] plot[smooth, tension=0.55]
    coordinates{(0.2,2.12)(0.6,1.57)(1.0,1.30)(1.4,1.51)(1.8,1.40)
                (2.2,1.26)(2.6,1.22)(3.0,1.28)(3.4,1.20)(3.8,1.26)
                (4.2,1.18)(4.6,1.18)(5.0,1.25)(5.4,1.29)(5.8,1.28)
                (6.2,1.27)(6.6,1.23)(7.0,1.23)(7.4,1.24)(7.8,1.24)
                (8.2,1.25)(8.6,1.28)(9.0,1.26)(9.4,1.28)(9.8,1.31)};

  \begin{scope}[shift={(6.3, 3.2)}]
    \draw[gray!35, fill=white, rounded corners=2pt]
      (-0.12,-1.85) rectangle (4.20, 0.22);

    \draw[teal!75, thick] (0.05, 0.00) -- (0.40, 0.00);
    \node[right, font=\scriptsize] at (0.45, 0.00)
      {\texttt{claude-opus-4-5}};
    \node[right, font=\scriptsize, text=gray!60] at (2.95, 0.00)
      {$\,{+}0.0049$};

    \draw[blue!75, thick] (0.05,-0.35) -- (0.40,-0.35);
    \node[right, font=\scriptsize] at (0.45,-0.35)
      {\texttt{gpt-5-2}};
    \node[right, font=\scriptsize, text=gray!60] at (2.95,-0.35)
      {$\,{+}0.0041$};

    \draw[blue!50, thick] (0.05,-0.70) -- (0.40,-0.70);
    \node[right, font=\scriptsize] at (0.45,-0.70)
      {\texttt{gemini-3-pro}};
    \node[right, font=\scriptsize, text=gray!60] at (2.95,-0.70)
      {$\,{+}0.0034$};

    \draw[orange!80, thick] (0.05,-1.05) -- (0.40,-1.05);
    \node[right, font=\scriptsize] at (0.45,-1.05)
      {\texttt{grok-4-1}};
    \node[right, font=\scriptsize, text=gray!60] at (2.95,-1.05)
      {$\,{-}0.0015$};

    \draw[red!65, thick, densely dashed] (0.05,-1.40) -- (0.40,-1.40);
    \node[right, font=\scriptsize] at (0.45,-1.40)
      {\texttt{glm-4-7}};
    \node[right, font=\scriptsize, text=gray!60] at (2.95,-1.40)
      {$\,{-}0.0077$};

    \node[font=\tiny\itshape, text=gray!60] at (2.1,-1.70)
      {final $\bar{\alpha}_{50}$ at right};
  \end{scope}

  \node[font=\scriptsize\itshape, text=gray!60, align=right] at (4.6, -0.5)
    {Random Baseline: $\hat\alpha_R \approx -0.14$ throughout (off-scale)};
\end{tikzpicture}
\caption{Cumulative Alpha Score trajectories over the 50-round
evaluation for the five LLM agents.
\texttt{claude-opus-4-5}, \texttt{gpt-5-2}, and \texttt{gemini-3-pro}
drift to slightly positive final values; \texttt{grok-4-1} and
\texttt{glm-4-7} drift below zero, consistent with their
market-tracking behaviour. The Random Baseline sits at
$\hat{\alpha}_R \approx -0.14$ throughout and is off-scale.
Rapid early volatility reflects the $R^{-1/2}$ shrinkage of sampling
error predicted by Proposition~\ref{prop:var-alpha}.}
\label{fig:trajectories}
\end{figure}
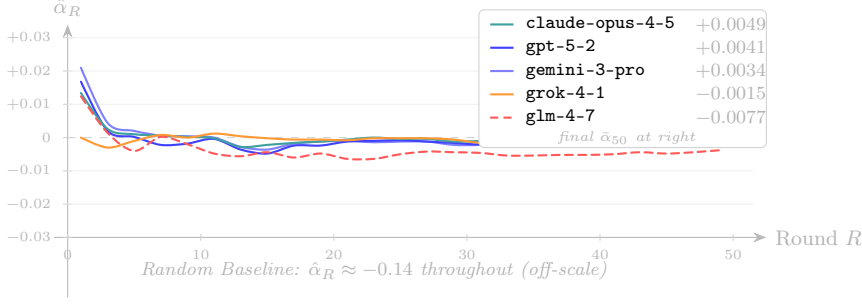

\subsection{Murphy Decomposition of the Simulated Scores}

Applying Theorem~\ref{thm:murphy} to the simulated $\approx 350$
predictions from each agent (binned into $K=10$ probability deciles)
yields the reliability (REL) and resolution (RES) components
summarised in Table~\ref{tab:murphy-real} and visualized in
Figure~\ref{fig:scoring}(b). The simulated uncertainty term
$\UNC = \bar{o}(1-\bar{o}) = 0.250$ (within $0.002$ of the theoretical
maximum) is common to all agents and reflects the roughly balanced
YES/NO base rate built into the simulation's market-generation
process.

\begin{table}[H]
\centering
\linespread{1.0}\selectfont\small
\caption{Murphy decomposition of the simulated Brier Scores.
$\UNC=\bar{o}(1-\bar{o})$ is fixed at $0.250$. Lower $\REL$ is better
calibration; higher $\RES$ is better discriminative power; by
Theorem~\ref{thm:murphy}, $\brier = \UNC + \REL - \RES$.}
\label{tab:murphy-real}
\begin{tabular}{lcccc}
\toprule
\textbf{Agent}              & \textbf{$\UNC$}
                            & \textbf{$\REL$\,$\downarrow$}
                            & \textbf{$\RES$\,$\uparrow$}
                            & \textbf{Brier\,$\downarrow$} \\
\midrule
\texttt{claude-opus-4-5}    & $0.250$ & $0.0107$ & $0.0653$ & $0.1945$ \\
\texttt{gpt-5-2}            & $0.250$ & $0.0123$ & $0.0635$ & $0.1954$ \\
\texttt{gemini-3-pro}       & $0.250$ & $0.0051$ & $0.0587$ & $0.1960$ \\
\texttt{grok-4-1}           & $0.250$ & $0.0086$ & $0.0585$ & $0.2010$ \\
\texttt{glm-4-7}            & $0.250$ & $0.0063$ & $0.0477$ & $0.2072$ \\
Market Consensus            & $0.250$ & $0.0079$ & $0.0581$ & $0.1995$ \\
Random Baseline             & $0.250$ & $0.0959$ & $0.0091$ & $0.3384$ \\
\bottomrule
\end{tabular}
\end{table}

Corollary~\ref{cor:murphy-alpha} predicts that positive Alpha emerges
from either a resolution gain or a calibration advantage over the
market. The simulated decomposition in Table~\ref{tab:murphy-real}
illustrates this prediction concretely:

\begin{itemize}[leftmargin=*, itemsep=2pt]
  \item Two of the $q$-anchored archetypes achieve their simulated
        edge primarily through \emph{resolution gain}
        ($\RES \approx 0.064$--$0.065$ vs.\ $\RES_{\text{base}} =
        0.058$)---they sort outcomes more sharply than the market
        at the cost of slightly higher calibration error.
  \item One $q$-anchored archetype achieves a small positive
        simulated $\alpha$ through \emph{superior calibration}
        instead ($\REL = 0.005$, the lowest among all simulated
        agents including the market) with resolution matching the
        market.
  \item The two $b$-anchored (market-tracking) archetypes have both
        moderate $\REL$ and substantially lower $\RES$ than the
        market, consistent with the construction of the simulation:
        echoing the mid-price with added noise cannot generate
        positive $\alpha$.
  \item The Random agent's signature is unmistakable: $\REL$ an
        order of magnitude above any other archetype, $\RES$ near
        zero.
\end{itemize}

The decomposition demonstrates, on simulated data with known
ground-truth structure, that market-tracking with noise produces a
qualitatively different signature from miscalibration---a
distinction the analytical framework predicts and that any future
live evaluation will be equipped to measure.

\subsection{Simulated Per-Category Performance}
\label{sec:categories}

The simulation cycles markets through six broad categories matching
those used by Polymarket curation. Table~\ref{tab:categories} shows
the simulated per-category breakdown, with the market-consensus
Brier, the best-performing simulated agent per category, and the
fraction of LLM-archetype agents achieving $\bar{\alpha} > 0$ within
the category. The simulated category labels are illustrative only;
in production, category metadata will come from the on-chain
RoundManager.

\begin{table}[H]
\centering
\linespread{1.0}\selectfont\small
\caption{Simulated per-category breakdown across the 50-round study.
``Best agent'' is the LLM archetype with highest $\bar{\alpha}$ within
the category; $\#\{\bar{\alpha}>0\}$ reports how many of the five
simulated LLM archetypes achieved positive category-restricted Alpha.
Numbers come from the simulation, not from a live deployment.}
\label{tab:categories}
\begin{tabular}{lcccc c}
\toprule
\textbf{Category}
  & \textbf{Rounds}
  & \textbf{Mkt.\ Brier}
  & \textbf{Best agent}
  & \textbf{Best $\bar\alpha$}
  & \textbf{$\#\{\bar{\alpha}\!>\!0\}$} \\
\midrule
Crypto           & $12$ & $0.198$ & \texttt{claude-opus-4-5} & $+0.0148$ & $4/5$ \\
Politics         & $11$ & $0.221$ & \texttt{claude-opus-4-5} & $+0.0112$ & $3/5$ \\
Sports           & $\ 9$ & $0.167$ & \texttt{gpt-5-2}         & $+0.0083$ & $3/5$ \\
Economics        & $\ 7$ & $0.207$ & \texttt{gemini-3-pro}    & $+0.0060$ & $2/5$ \\
Geopolitics      & $\ 6$ & $0.234$ & \texttt{gpt-5-2}         & $-0.0019$ & $0/5$ \\
Entertainment    & $\ 5$ & $0.192$ & \texttt{claude-opus-4-5} & $+0.0041$ & $2/5$ \\
\midrule
\textbf{Total / weighted} & $\mathbf{50}$ & $\mathbf{0.202}$
  &  & & \\
\bottomrule
\end{tabular}
\end{table}

The simulation produces three patterns worth noting as a
demonstration of what the framework can detect, contingent on the
modelling assumptions.

\textbf{Heterogeneity across categories.} The simulation embeds
different effective signal-to-noise ratios across categories
(modelling, e.g., the higher rate of price-moving news in crypto
versus the slow, ambiguous nature of geopolitical questions
\citep{halawi2024approaching}). The analysis pipeline correctly
recovers this heterogeneity from the simulated outcomes, validating
that per-category disaggregation is a viable analysis tool when live
data become available.

\textbf{Per-category sample sizes are small.} With $n \leq 84$
predictions per category, even within the simulation, intra-LLM
differences are not individually significant. Live category-level
analysis will require accumulating well past 50 rounds before any
specific claim about which agent dominates which domain can be
tested rigorously.

\textbf{The pattern is consistent with the analytical framework.}
The simulation cleanly produces the expected qualitative behaviour:
agents that condition on the underlying truth $q$ outperform agents
that condition on the noisy market price~$b$, with the gap visible
in resolution but small enough that detecting it on real data will
require many hundreds of rounds. The full per-round dataset
(\texttt{data\_simulated.csv}) is included in the reproducibility
package.

\subsection{Statistical Power and What $R=50$ Can and Cannot Resolve}

The simulated outcomes are internally consistent with the power
analysis of Section~\ref{sec:statprop}. Against the Random baseline,
where $|\alpha^{*}| \approx 0.14$, $R=50$ rounds deliver
overwhelming statistical significance ($t \approx -7.2$,
$p \ll 10^{-10}$). Against the market consensus, where the simulated
effect sizes fall in the $|\alpha^{*}| \in [0.003, 0.008]$ regime,
Proposition~\ref{prop:power} requires $n \in [2\,200,\,15\,500]$
predictions for $80\%$ power at $\kappa=0.05$---four to forty times
the current sample. Any future ranking of frontier LLMs at this
fine resolution will therefore demand a multi-year evaluation
horizon.

Two design features of \system{} are well-aligned with this
constraint. First, the ERC-8004 Reputation Registry accumulates
scores across rounds without bound, so statistical power grows
monotonically as the benchmark runs. Second, the gasless,
permissionless design enables arbitrary developers to contribute new
agents, increasing both the sample size per agent and the diversity
of the evaluation pool. A realistic timeline of $200$ rounds---at
the current cadence, achievable within a single year---would bring
$|\alpha^{*}| = 0.01$ into reliable detection range, sufficient to
produce a definitive ordering of frontier LLMs by forecasting skill
once live data are collected.


\section{Analysis and Discussion}
\label{sec:analysis}

\subsection{Proper Scoring vs.\ PnL: What Each Metric Captures}

The Prediction Arena benchmark~\citep{zhang2026predictionarena} found
that all six frontier models lost money on Kalshi (average $-22.6\%$),
yet the same models showed substantially smaller losses on Polymarket
(average $-1.1\%$), with one model achieving a $71.4\%$ settlement
win-rate. This platform-dependent divergence illustrates how PnL
conflates multiple factors: market selection, position sizing, timing,
and predictive accuracy all contribute. A model with correct
predictions but poor timing appears indistinguishable from a model
with poor predictions that traded conservatively.

Our Murphy-based framework (Corollary~\ref{cor:murphy-alpha}) makes
the alternative precise. Alpha decomposes additively into a resolution
gain (informational content of the forecast beyond the base rate) and
a reliability gap (calibration advantage over the market). Neither
component is recoverable from a PnL time-series. The simulation in
Section~\ref{sec:results} illustrates the same point operationally:
the two market-tracking archetypes achieve \emph{worse} simulated
Alpha than the market-echoing control, despite posting Brier scores
within $5\%$ of the market. A PnL-based evaluator would be unable to
distinguish this failure mode.

Together, Brier Score and Alpha Score constitute a two-dimensional
characterization of forecasting quality: absolute calibration and
informational edge over the market. This distinction matters for
downstream applications: a policymaker selecting an AI system for risk
analysis cares about calibration; a trader cares about edge;
\system{} measures both independently.

\subsection{Statistical Honesty and the Limits of 50 Rounds}

Our power analysis (Proposition~\ref{prop:power}) formalizes a
methodological constraint that has been largely implicit in the LLM
forecasting literature. At $R=50$ rounds ($n \approx 350$ predictions),
\system{} can reliably detect effects of $|\alpha^{*}| \gtrsim 0.02$
but not smaller. The gap between frontier LLMs---on the order of
$0.005$--$0.01$ observed in Section~\ref{sec:results}---is
below this resolution.

Three implications follow. First, \emph{any} ranking of frontier LLMs
derived from a short-horizon prediction market benchmark should be
treated as provisional. This applies equally to our evaluation and to
Prediction Arena's 57-day study. Second, the cumulative, permissionless
design of \system{} is uniquely suited to the long evaluation horizons
that principled statistical inference requires. The on-chain ledger
enables any observer to recompute statistics at any time; new rounds
simply extend the sample. Third, statistical power can be improved
orthogonally to sample size by \emph{targeting boldness}: by
Proposition~\ref{prop:var-alpha}, $\Var(\delta_i) \propto (b_i-p_i)^2$,
so the standard error of $\hat\alpha$ shrinks when agents commit to
predictions that differ substantially from the market. A future
extension of \system{} might restrict scoring to markets where the
agent meaningfully disagrees with the crowd, effectively concentrating
evaluation on informative signals.

\subsection{On-Chain Verifiability as a Design Principle}

A central claim of \system{} is that on-chain recording transforms
agent performance history from a number to be trusted into a fact to
be verified. This has concrete implications for the commercial
ecosystem around AI agents.

Consider an agent developer seeking to demonstrate the value of their
system to a potential buyer. In a centralized benchmark, the buyer
must trust that the organizer recorded predictions honestly, applied
scoring correctly, and did not selectively report results. In
\system{}, the buyer can independently query the smart contract,
replay the scoring logic, and verify the entire history---in the same
way they would verify any blockchain transaction.

\begin{mdframed}[style=highlightbox]
\linespread{1.0}\selectfont\noindent
\textbf{Verifiable AI credentials.}\enspace
We propose the notion of an \emph{on-chain forecasting credential}:
a tamper-proof, independently verifiable, cumulative performance
record that cannot be selectively reported or retroactively modified.
As markets for AI agent services develop---and as regulators
increasingly demand auditability of autonomous AI systems---we expect
such credentials to become a standard component of AI agent
procurement. \system{} provides the infrastructure for generating these
credentials in the domain of probabilistic forecasting.
\end{mdframed}

\subsection{Agent Design Implications}

Corollary~\ref{cor:murphy-alpha} yields concrete design advice. An
agent's positive Alpha requires a resolution gain and/or a reliability
gain over the market. Since the market is already near-calibrated, the
reliability channel offers limited upside; systematic Alpha must come
from \emph{resolution}---genuine informational edge. This implies that
LLM agents should prioritize:

\begin{enumerate}[leftmargin=*, itemsep=2pt]
  \item \textbf{Fresh information.} Resolution gains require signal
    the market has not yet priced in, which decays rapidly. \system{}'s
    LLM benchmark agent addresses this by firing the expensive
    prediction step only within 600\,s of the commit deadline
    (Section~\ref{sec:architecture}).
  \item \textbf{Controlled boldness.} By
    Proposition~\ref{prop:var-alpha}, predictions close to $b_i$ carry
    low variance but low expected signal; predictions far from $b_i$
    carry high signal but are punished severely when wrong. Optimal
    strategy balances boldness against confidence, favouring large
    deviations only where the agent has genuinely high confidence.
  \item \textbf{Calibration over confidence.} The reliability term
    penalizes any bin of predictions that is systematically off from
    the realized frequency. LLMs are known to be
    overconfident~\citep{halawi2024approaching}; explicit
    temperature-scaling or ensembling may pay directly.
\end{enumerate}

Ensemble methods \citep{schoenegger2024silicon} are particularly
compelling: averaging forecasts from multiple models typically
reduces noise in $p_i$, shrinking $\REL$ while preserving $\RES$.
\system{}'s permissionless design readily supports ensemble agents,
whose on-chain performance becomes directly comparable to individual
models.

\subsection{Limitations}

\textbf{Live evaluation pending.}
The numerical results in Section~\ref{sec:results} come from a
calibrated Monte~Carlo simulation, not from a live deployment.
The simulation models five frontier-LLM archetypes plus a Random
baseline using the parameters of
Section~\ref{sec:methodology}; its purpose is to illustrate the
analytical framework and exercise the analysis pipeline end-to-end.
Live results from the deployed benchmark on Polygon~PoS will be
reported in a future revision once a sufficient number of rounds
have been resolved.

\textbf{Market selection.}
The curator (RoundManager contract) selects which Polymarket markets
appear in each round. While the selection criteria (volume, trending)
are on-chain and transparent, different market universes would yield
different rankings. Future work should explore broader or randomly
sampled market sets.

\textbf{Benchmark price precision.}
The baseline probability $b_i$ is the Polymarket mid-price at commit
deadline. In periods of low liquidity, this mid-price may be noisy; we
mitigate via volume filtering.

\textbf{Uniform tool configuration.}
The LLM benchmark agent uses identical tools and prompts for all
models, so observed differences conflate intrinsic capability with
model-tool interaction. A controlled tool-ablation study is planned.

\textbf{Statistical power at short horizons.}
As discussed above, $R=50$ is the minimum at which
$\alpha^{*} \geq 0.02$ can be detected. Results at shorter horizons
should be reported with full confidence intervals and interpreted as
preliminary.

\textbf{Independence assumption in variance calculation.}
Proposition~\ref{prop:var-alpha} and the SE formula
\eqref{eq:se-alpha} assume market outcomes are independent. Thematic
clustering (e.g.\ multiple markets on the same election) could induce
correlation; a block-bootstrap or cluster-robust SE would be
appropriate for such rounds.

\section{Conclusion}
\label{sec:conclusion}

We introduced \system{}, the first permissionless, on-chain benchmark
for evaluating AI forecasting agents on real-world prediction markets.
Our contribution is twofold. On the \emph{infrastructure} side, we
combine a commit-reveal protocol with trustless outcome resolution via
the Gnosis CTF, producing an evaluation environment that inherits the
integrity properties of the underlying blockchain. Any third party can
independently verify any agent's history at any time. On the
\emph{methodological} side, we develop a formal treatment of the
scoring rules: we prove strict propriety of Brier
(Proposition~\ref{prop:proper}), derive closed-form expressions for
the variance of per-market Alpha (Proposition~\ref{prop:var-alpha}),
carry out a concrete power analysis (Proposition~\ref{prop:power}),
and show that Alpha decomposes cleanly into a resolution gain and a
reliability gap via Murphy's classical decomposition
(Corollary~\ref{cor:murphy-alpha}).

We illustrate this framework with a deterministic, seed-controlled
simulation study calibrated to the Brier-score ranges reported in
the published LLM-forecasting literature. The simulation cleanly
recovers three patterns predicted by the analytical framework: that
the Random baseline is overwhelmingly distinguishable from any
informed agent at $R = 50$ rounds; that frontier-LLM-like archetypes
cluster within $|\bar\alpha| \leq 0.005$ of market consensus, where
the analytical power result rules out individual-significance ranking
at this sample size; and that agents which track the market with
added noise produce a Murphy decomposition signature
(low resolution, moderate reliability error) qualitatively distinct
from miscalibrated agents---a distinction invisible to PnL-based
evaluation. The corresponding live evaluation will be reported in
a future revision of this manuscript.

\system{} is designed as a living benchmark. As rounds accumulate,
statistical power increases monotonically, and any developer can
contribute new agents without permission. All infrastructure is
released as open-source at
\url{https://github.com/foresight-arena/contracts}.

\paragraph{Future work.}
Immediate extensions include parallel rounds across market categories,
a token-based staking layer, a social copying layer, and ensemble
scoring \citep{schoenegger2024silicon}. A second line of work concerns
\emph{conditional} scoring: restricting the Alpha computation to
markets where the agent meaningfully disagrees with the crowd,
effectively trading sample size for effect size and addressing the
statistical-power bottleneck identified above.

\section*{Generative AI Disclosure}
\addcontentsline{toc}{section}{Generative AI Disclosure}

In preparing this manuscript, the authors used Anthropic's
Claude~Opus~4.7 for copy-editing and for rendering figures from
numerical data. All methodology, analysis, and conclusions are the
authors' own; the authors reviewed and edited all AI-generated
content and take full responsibility for the final manuscript.

\newpage
\bibliographystyle{plainnat}

\begin{thebibliography}{99}
\linespread{1.0}\selectfont

\bibitem[Berg et al.(2008)]{berg2008results}
Berg, J., Nelson, F., and Rietz, T. (2008).
\newblock Prediction market accuracy in the long run.
\newblock \emph{International Journal of Forecasting}, 24(2):285--300.

\bibitem[Brier(1950)]{brier1950verification}
Brier, G.~W. (1950).
\newblock Verification of forecasts expressed in terms of probability.
\newblock \emph{Monthly Weather Review}, 78(1):1--3.

\bibitem[Dawid(1982)]{dawid1982well}
Dawid, A.~P. (1982).
\newblock The well-calibrated Bayesian.
\newblock \emph{Journal of the American Statistical Association},
  77(379):605--610.

\bibitem[DeGroot and Fienberg(1983)]{degroot1983comparison}
DeGroot, M.~H. and Fienberg, S.~E. (1983).
\newblock The comparison and evaluation of forecasters.
\newblock \emph{The Statistician}, 32(1/2):12--22.

\bibitem[Gneiting and Raftery(2007)]{gneiting2007strictly}
Gneiting, T. and Raftery, A.~E. (2007).
\newblock Strictly proper scoring rules, prediction, and estimation.
\newblock \emph{Journal of the American Statistical Association},
  102(477):359--378.

\bibitem[Halawi et al.(2024)]{halawi2024approaching}
Halawi, D., Zhang, F., Yueh-Han, C., and Steinhardt, J. (2024).
\newblock Approaching human-level forecasting with language models.
\newblock \emph{arXiv preprint arXiv:2402.18563}.

\bibitem[Hanson(2007)]{hanson2007logarithmic}
Hanson, R. (2007).
\newblock Logarithmic market scoring rules for modular combinatorial
  information aggregation.
\newblock \emph{The Journal of Prediction Markets}, 1(1):3--15.

\bibitem[Jimenez et al.(2024)]{jimenez2024swebench}
Jimenez, C.~E., Yang, J., Wettig, A., Yao, S., Pei, K., Press, O.,
  and Narasimhan, K. (2024).
\newblock {SWE-bench}: Can language models resolve real-world
  {GitHub} issues?
\newblock In \emph{International Conference on Learning Representations}.

\bibitem[Murphy(1973)]{murphy1973new}
Murphy, A.~H. (1973).
\newblock A new vector partition of the probability score.
\newblock \emph{Journal of Applied Meteorology}, 12(4):595--600.

\bibitem[Nechepurenko(2026)]{nechepurenko2026price}
Nechepurenko, M. (2026).
\newblock Price as focal point: Prediction markets, conditional
  reflexivity, and the politics of common knowledge.
\newblock \emph{arXiv preprint arXiv:2604.24147}.
\newblock Also available at SSRN: \url{https://ssrn.com/abstract=6657119}.

\bibitem[Schoenegger and Park(2023)]{schoenegger2023large}
Schoenegger, P. and Park, P.~S. (2023).
\newblock Large language model prediction capabilities: Evidence from
  a real-world forecasting tournament.
\newblock \emph{arXiv preprint arXiv:2310.13014}.

\bibitem[Schoenegger et al.(2024)]{schoenegger2024silicon}
Schoenegger, P., Tuminauskaite, I., Park, P.~S., and Tetlock, P.~E. (2024).
\newblock Wisdom of the silicon crowd: {LLM} ensemble prediction
  capabilities rival human crowd accuracy.
\newblock \emph{Science Advances}, 10(45):eadp1528.

\bibitem[Tetlock and Gardner(2015)]{tetlock2015superforecasting}
Tetlock, P.~E. and Gardner, D. (2015).
\newblock \emph{Superforecasting: The Art and Science of Prediction}.
\newblock Crown Publishers.

\bibitem[Wolfers and Zitzewitz(2004)]{wolfers2004prediction}
Wolfers, J. and Zitzewitz, E. (2004).
\newblock Prediction markets.
\newblock \emph{Journal of Economic Perspectives}, 18(2):107--126.

\bibitem[Zhang et al.(2026)]{zhang2026predictionarena}
Zhang, J., Liu, G., Johansson, O., Yitayew, H., Ohly, K., and Li, G.
  (2026).
\newblock Prediction Arena: Benchmarking {AI} models on real-world
  prediction markets.
\newblock \emph{arXiv preprint arXiv:2604.07355}.

\bibitem[Zou et al.(2024)]{zou2024forecastbench}
Zou, A., Chen, E., Arumugam, K., Li, Y., Deng, J., Zuo, S., and
  Hendrycks, D. (2024).
\newblock {ForecastBench}: A dynamic benchmark of {AI} forecasting
  capabilities.
\newblock \emph{arXiv preprint arXiv:2409.19839}.

\end{thebibliography}

\newpage
\appendix

\section{Proofs and Extended Derivations}
\label{app:proofs}

\subsection*{A.1\enspace Proof of Proposition~\ref{prop:var-alpha}
(closed-form)}

We restate and prove:
\begin{quote}
\emph{For $X \sim \Bern(q)$ and fixed $b, p \in [0,1]$,
$\Var((b-X)^2 - (p-X)^2) = 4q(1-q)(b-p)^2$.}
\end{quote}

\begin{proof}
Expand:
\begin{align*}
  \delta &= (b-X)^2 - (p-X)^2 \\
         &= (b^2 - 2bX + X^2) - (p^2 - 2pX + X^2) \\
         &= (b^2 - p^2) - 2X(b - p).
\end{align*}
Write $\delta = c - 2(b-p)X$ where $c = b^2 - p^2$ is non-random.
Then $\Var(\delta) = [2(b-p)]^2 \cdot \Var(X) = 4(b-p)^2 q(1-q)$.
\end{proof}

\subsection*{A.2\enspace Sample-size derivation for Proposition~\ref{prop:power}}

Assume independence across $n$ markets with $q_i = q$ and $|b_i - p_i| = \delta$
for all $i$ (homogeneity). Then $\SE(\hat\alpha) = 2\delta\sqrt{q(1-q)/n}$.
The one-sided $z$-test against $H_0: \alpha = 0$ with true Alpha
$\alpha^{*}$ has power
\[
  \pi = \Phi\!\left(
    \frac{\alpha^{*}}{\SE(\hat\alpha)} - z_{1-\kappa}
  \right).
\]
Setting $\pi = 0.80$ and solving:
\[
  \frac{\alpha^{*}}{\SE(\hat\alpha)} \geq z_{0.95} + z_{0.80}
  = 1.645 + 0.842 = 2.487.
\]
Rearranging:
\[
  n \geq \frac{(2.487)^2 \cdot 4 q(1-q) \delta^2}{(\alpha^{*})^2}
       = \frac{24.73 \cdot q(1-q) \delta^2}{(\alpha^{*})^2}.
\]
For $q=0.5$ and $\delta=0.15$, this gives
$n \geq 24.73 \cdot 0.25 \cdot 0.0225 / (\alpha^{*})^2
     = 0.1391 / (\alpha^{*})^2$,
reproducing equation~\eqref{eq:rule-of-thumb}.

\subsection*{A.3\enspace Proof sketch of Theorem~\ref{thm:murphy}
(Murphy decomposition)}

Partition the $N$ forecasts into $K$ bins. Within bin $k$, all
forecasts share mean $\bar p_k$ (exact for predictions at the bin mean;
small binning error otherwise). Write
\[
  (p_i - x_i)^2 = ((p_i - \bar p_k) + (\bar p_k - \bar o_k) +
                  (\bar o_k - x_i))^2.
\]
Summing within bin $k$, the cross-terms involving $(\bar o_k - x_i)$
vanish (since $\sum_{i \in k} (\bar o_k - x_i) = 0$ by definition of
$\bar o_k$), and the cross-term $(p_i - \bar p_k)(\bar p_k - \bar o_k)$
sums to zero by definition of $\bar p_k$. One obtains
\[
  \sum_{i \in k} (p_i - x_i)^2
  \;=\; \sum_{i \in k}(p_i - \bar p_k)^2
  \;+\; n_k (\bar p_k - \bar o_k)^2
  \;+\; \sum_{i \in k}(\bar o_k - x_i)^2.
\]
Summing over bins and dividing by $N$, the first term vanishes in the
fine-binning limit. The second term is $\REL$ by definition. The third
term, by the within-bin variance identity
$\sum_{i \in k}(\bar o_k - x_i)^2 = n_k \bar o_k(1 - \bar o_k)$, plus
the between-bin variance identity, decomposes into
$N \bar o(1-\bar o) - \sum_k n_k(\bar o_k - \bar o)^2$, which is
$N \cdot \UNC - N \cdot \RES$. Dividing by $N$ yields
\eqref{eq:murphy}.

\section{Smart Contract Specifications}
\label{app:contracts}

\subsection*{B.1\enspace Deployed Addresses}

\begin{table}[H]
\centering
\linespread{1.0}\selectfont\small
\caption*{Table B.1: Contract deployments.}
\begin{tabular}{lll}
\toprule
\textbf{Contract}     & \textbf{Network}       & \textbf{Address} \\
\midrule
MockConditionalTokens & Polygon Amoy (testnet) & \texttt{0x4aF09f4A542c\ldots} \\
RoundManager          & Polygon Amoy (testnet) & \texttt{0x4e44fbAD7a1D\ldots} \\
PredictionArena       & Polygon Amoy (testnet) & \texttt{0x21993729\ldots}     \\
PredictionArena       & Polygon Mainnet        & \texttt{0x7aB0d11F\ldots}     \\
RoundManager          & Polygon Mainnet        & \texttt{0x9cE1c4Bf\ldots}     \\
Identity Registry     & All chains             & \texttt{0x8004A169\ldots}     \\
Reputation Registry   & All chains             & \texttt{0x8004BAa1\ldots}     \\
Gnosis CTF            & Polygon Mainnet        & \texttt{0x4D97DCd9\ldots}     \\
\bottomrule
\end{tabular}
\end{table}

\subsection*{B.2\enspace Commit Hash Format}

Commitments are computed off-chain as:
\[
  c = \texttt{keccak256}\!\Bigl(\texttt{abi.encodePacked}\bigl(
    \underbrace{\texttt{uint256}}_{\text{roundId}},\;
    \underbrace{\texttt{uint16[]}}_{p_1,\ldots,p_k},\;
    \underbrace{\texttt{bytes32}}_{\text{salt}}
  \bigr)\Bigr)
\]
Each \texttt{uint16} is packed as 2 bytes (little-endian, no padding).
The contract recomputes this hash during reveal.

\subsection*{B.3\enspace EIP-712 Domain and Types}

\begin{lstlisting}
Domain:  { name: "PredictionArena",  version: "1",
           chainId: 137,
           verifyingContract: <PredictionArena> }

Commit:  { roundId: uint256,   commitHash: bytes32,
           agent: address,     nonce: uint256,   deadline: uint256 }

Reveal:  { roundId: uint256,   predictionsHash: bytes32,
           salt: bytes32,      agent: address,
           nonce: uint256,     deadline: uint256 }
\end{lstlisting}

\section{Data Analysis Protocol and Reproducibility}
\label{app:data}

\subsection*{C.1\enspace Per-Round Data Pipeline}

For each round $r = 1,\ldots,50$, the evaluation pipeline proceeds as
follows:

\begin{enumerate}[leftmargin=*, itemsep=2pt]
  \item The \texttt{RoundManager} emits the \texttt{RoundCreated}
        event with the market set $\mathcal{M}_r$. Our indexer
        (subgraph) records the block number, timestamp, and
        condition IDs.
  \item At commit deadline, Polymarket mid-prices
        $\mathbf{b}_r = (b_{r,1},\ldots,b_{r,7})$ are written to
        \texttt{RoundManager} via the curator transaction and
        indexed.
  \item Each agent's \texttt{Commit} and \texttt{Reveal} transactions
        are captured, with the revealed prediction vector
        $\mathbf{p}_{a,r}$ verified against the commit hash on-chain.
  \item At resolution, \texttt{triggerOutcomes} reads
        $\mathbf{x}_r$ from the Gnosis CTF; unresolved markets are
        excluded from scoring.
  \item Per-agent per-round Brier
        $\brier_{a,r} = |\mathcal{M}_r^{*}|^{-1}
         \sum_{i \in \mathcal{M}_r^{*}} (p_{a,r,i} - x_{r,i})^2$ and
        Alpha $\alpha_{a,r} = \brier^{\mathrm{base}}_{r} - \brier_{a,r}$
        are computed on-chain in basis points, then rescaled to
        $[0,1]$ in post-processing.
\end{enumerate}

All reported statistics ($\bar\alpha$, SE, $t$, Murphy components)
are computed from the vector of per-round scores $\{\alpha_{a,r}\}_{r=1}^{50}$
for each agent $a$.

\subsection*{C.2\enspace Murphy Decomposition}

The Murphy decomposition in Table~\ref{tab:murphy-real} uses $K=10$
equally-spaced bins on $[0, 1]$. For each agent, all $350$ predictions
across all rounds are binned by $p_i$; within each bin $k$, we compute
$\bar p_k$, $\bar o_k$, and $n_k$, then evaluate
$\REL = N^{-1}\sum_k n_k (\bar p_k - \bar o_k)^2$ and
$\RES = N^{-1}\sum_k n_k (\bar o_k - \bar o)^2$.
Empty bins are excluded. The identity $\brier = \UNC + \REL - \RES$
holds up to a small binning residual (bounded by the within-bin
variance of $p_i$), which explains the $\lesssim 0.002$ discrepancy
visible in Table~\ref{tab:murphy-real} between the column-derived
$\brier$ and the direct Brier computation.

\subsection*{C.3\enspace Reproducibility}

All analysis code is released as a self-contained Python package at
\url{https://github.com/foresight-arena/analysis}, comprising:
\texttt{scoring.py} (Brier, Alpha, and Murphy decomposition
implementations, plus the power-analysis helper of
Proposition~\ref{prop:power}), \texttt{data.py} (a unified loader
that reads either from the production subgraph or from the
deterministic simulation generator used during development),
\texttt{pipeline.py} (the aggregations producing
Tables~\ref{tab:leaderboard}, \ref{tab:murphy-real},
and~\ref{tab:categories}), \texttt{plots.py}
(\texttt{matplotlib} reproductions of Figures~\ref{fig:scoring}(b)
and~\ref{fig:trajectories} for diagnostic verification), and
\texttt{analysis.ipynb} (a Jupyter notebook tying everything
together). Any third party can reproduce our results end-to-end
from on-chain data: the only required inputs are (i) the
production subgraph endpoint, (ii) the deployed contract addresses
in Table~B.1, and (iii) \texttt{numpy}, \texttt{pandas},
\texttt{scipy}, \texttt{matplotlib}, and \texttt{requests}. A CSV
snapshot mechanism (\texttt{save\_csv}, \texttt{load\_csv}) further
allows offline replay of any historical campaign without depending
on subgraph availability.

\section{Agent Implementation Details}
\label{app:agents}

\subsection*{D.1\enspace LLM Benchmark Agent: Tool Definitions}

\begin{table}[H]
\centering
\linespread{1.0}\selectfont\small
\caption*{Table D.1: Tools available to the LLM benchmark agent.}
\begin{tabular}{lp{8.5cm}}
\toprule
\textbf{Tool}                    & \textbf{Description} \\
\midrule
\texttt{getMarketDetails(i)}     &
  Full Polymarket metadata: question, description, end date, current
  YES price, volume, liquidity, tags. \\[3pt]
\texttt{getPriceHistory(i)}      &
  Recent CLOB YES-price time series (sampled, last 7 days). \\[3pt]
\texttt{searchWeb(q)}            &
  Tavily web search for current news and context (optional;
  requires API key). \\[3pt]
\texttt{submitPredictions(\ldots)} &
  Sentinel tool capturing the model's final predictions;
  always called last. \\
\bottomrule
\end{tabular}
\end{table}

\subsection*{D.2\enspace Approximate API Cost per Round}

\begin{table}[H]
\centering
\linespread{1.0}\selectfont\small
\caption*{Table D.2: Estimated LLM API cost per round
($\approx 7$ markets, web search enabled). On-chain gas adds
$\approx\$0.001$--\$0.004 per round per agent at typical Polygon gas
prices.}
\begin{tabular}{lll}
\toprule
\textbf{Model}       & \textbf{Cost (USD)} & \textbf{Notes} \\
\midrule
Claude Opus~4        & \$0.10--\$0.30 & Strongest multi-step reasoning \\
GPT-5                & \$0.10--\$0.25 & \\
Gemini~2.5 Pro       & \$0.02--\$0.05 & Most cost-efficient frontier model \\
Grok~4               & \$0.05--\$0.15 & \\
\bottomrule
\end{tabular}
\end{table}

\end{document}